\documentclass[11pt,onecolumn]{article}%
\usepackage{amsmath}
\usepackage{amsfonts}
\usepackage{amssymb}
\usepackage{graphicx}
\usepackage{fullpage}%
\providecommand{\U}[1]{\protect\rule{.1in}{.1in}}
\newtheorem{theorem}{Theorem}

\newtheorem{conjecture}[theorem]{Conjecture}
\newtheorem{corollary}[theorem]{Corollary}

\newtheorem{definition}[theorem]{Definition}

\newtheorem{lemma}[theorem]{Lemma}

\newenvironment{proof}[1][Proof]{\noindent\textbf{#1.} }{\ \rule{0.5em}{0.5em}}
\begin{document}

\title{BQP and the Polynomial Hierarchy}
\author{Scott Aaronson\thanks{MIT. \ Email: aaronson@csail.mit.edu. \ Supported by an
NSF CAREER Award, a DARPA YFA grant, MIT CSAIL, and the Keck Foundation.}}
\date{}
\maketitle

\begin{abstract}
The relationship between\ $\mathsf{BQP}$\ and $\mathsf{PH}$ has been an open
problem since the earliest days of quantum computing. \ We present evidence
that quantum computers can solve problems outside the entire polynomial
hierarchy, by relating this question to topics in circuit complexity,
pseudorandomness, and Fourier analysis.

First, we show that there exists an oracle relation problem (i.e., a problem
with many valid outputs) that is solvable in $\mathsf{BQP}$, but not in
$\mathsf{PH}$. \ This also yields a non-oracle relation problem that is
solvable in quantum \textit{logarithmic} time, but not in $\mathsf{AC}^{0}$.

Second, we show that an oracle \textit{decision }problem separating
$\mathsf{BQP}$\ from $\mathsf{PH}$ would follow from the \textit{Generalized
Linial-Nisan Conjecture}, which we formulate here and which is likely of
independent interest. \ The original Linial-Nisan Conjecture\ (about
pseudorandomness against constant-depth circuits)\ was recently proved by
Braverman, after being open for twenty years.

\end{abstract}
\tableofcontents

\section{Introduction\label{INTRO}}

A central task of quantum computing theory is to understand how $\mathsf{BQP}%
$---meaning Bounded-Error Quantum Polynomial-Time, the class of all problems
feasible for a quantum computer---fits in with classical complexity classes.
\ In their original 1993 paper defining $\mathsf{BQP}$, Bernstein and Vazirani
\cite{bv}\ showed that $\mathsf{BPP}\subseteq\mathsf{BQP}\subseteq
\mathsf{P}^{\mathsf{\#P}}$.\footnote{The upper bound was later improved to
$\mathsf{BQP}\subseteq\mathsf{PP}$\ by Adleman, DeMarrais, and Huang
\cite{adh}.} \ Informally, this says that quantum computers are at least as
fast as classical probabilistic computers and no more than exponentially
faster (indeed, they can be simulated using an oracle for counting).
\ Bernstein and Vazirani also gave evidence that $\mathsf{BPP}\neq
\mathsf{BQP}$, by exhibiting an oracle problem called \textsc{Recursive
Fourier Sampling} that requires $n^{\Omega\left(  \log n\right)  }$ queries on
a classical computer but only $n$ queries on a quantum computer.\footnote{For
more about \textsc{Recursive Fourier Sampling} see Aaronson \cite{aar:rfs}.}
\ The evidence for the power of quantum computers became dramatically stronger
a year later, when Shor \cite{shor} (building on work of Simon \cite{simon})
showed that \textsc{Factoring} and \textsc{Discrete Logarithm} are in
$\mathsf{BQP}$. \ On the other hand, Bennett et al.\ \cite{bbbv}\ gave oracle
evidence that $\mathsf{NP}\not \subset \mathsf{BQP}$,\ and while no one
regards such evidence as decisive, today it seems extremely unlikely that
quantum computers can solve $\mathsf{NP}$-complete problems in polynomial
time. \ A vast body of research, continuing to the present, has sought to map
out the detailed boundary between those\ $\mathsf{NP}$\ problems that are
feasible for quantum computers and those that are not.

However, there is a complementary question that---despite being universally
recognized as one of the \textquotedblleft grand challenges\textquotedblright%
\ of the field---has had essentially zero progress over the last sixteen years:

\begin{quotation}
\noindent\textit{Is }$\mathsf{BQP}$\textit{ in }$\mathsf{NP}$\textit{? \ More
generally, is }$\mathsf{BQP}$\textit{ contained anywhere in the polynomial
hierarchy}$\ \mathsf{PH}=\mathsf{NP}\cup\mathsf{NP}^{\mathsf{NP}}%
\cup\mathsf{NP}^{\mathsf{NP}^{\mathsf{NP}}}\cup\cdots$\textit{?}
\end{quotation}

The \textquotedblleft default\textquotedblright\ conjecture is presumably
$\mathsf{BQP}\not \subset \mathsf{PH}$, since no one knows what a simulation
of $\mathsf{BQP}$\ in $\mathsf{PH}$\ would look like. \ Before this work,
however, there was no formal evidence for or against that conjecture. \ Almost
all the problems for which we have quantum algorithms---including
\textsc{Factoring} and \textsc{Discrete Logarithm}---are easily seen to be in
$\mathsf{NP}\cap\mathsf{coNP}$.\footnote{Here we exclude $\mathsf{BQP}%
$-complete problems such as approximating the Jones polynomial \cite{ajl},
which, by the very fact of being $\mathsf{BQP}$-complete, seem hard to
interpret as \textquotedblleft evidence\textquotedblright\ for $\mathsf{BQP}%
\not \subset \mathsf{PH}$.} \ One notable exception is \textsc{Recursive
Fourier Sampling}, the problem that Bernstein and Vazirani \cite{bv}%
\ originally used to construct an oracle $A$ relative to which $\mathsf{BPP}%
^{A}\neq\mathsf{BQP}^{A}$. \ One can show, without too much difficulty, that
\textsc{Recursive Fourier Sampling}\ yields oracles $A$ relative to which
$\mathsf{BQP}^{A}\not \subset \mathsf{NP}^{A}$ and indeed $\mathsf{BQP}%
^{A}\not \subset \mathsf{MA}^{A}$. \ However, while it is reasonable to
conjecture that \textsc{Recursive Fourier Sampling}\ (as an oracle problem) is
not in $\mathsf{PH}$, it is open even to show that this problem (or any other
$\mathsf{BQP}$\ oracle problem) is not in $\mathsf{AM}$! \ Recall that
$\mathsf{AM}=\mathsf{NP}$\ under plausible derandomization assumptions
\cite{kvm}. \ Thus, until we solve the problem of constructing an oracle $A$
such that $\mathsf{BQP}^{A}\not \subset \mathsf{AM}^{A}$, we cannot even claim
to have oracle evidence (which is itself, of course, a weak form of evidence)
that $\mathsf{BQP}\not \subset \mathsf{NP}$.

Before going further, we should clarify that there are two questions here:
whether $\mathsf{BQP}\subseteq\mathsf{PH}$ and whether $\mathsf{P{}%
romiseBQP}\subseteq\mathsf{P{}romisePH}$. \ In the unrelativized world, it is
entirely possible that quantum computers can solve promise problems outside
the polynomial hierarchy, but that all \textit{languages} in $\mathsf{BQP}%
$\ are\ nevertheless in $\mathsf{PH}$. \ However, for the specific purpose of
constructing an oracle $A$ such that $\mathsf{BQP}^{A}\not \subset
\mathsf{PH}^{A}$, the two questions are equivalent, basically because one can
always \textquotedblleft offload\textquotedblright\ a promise into the
construction of the oracle $A$.\footnote{Here is a simple proof: let
$\Pi=\left(  \Pi_{\operatorname*{YES}},\Pi_{\operatorname*{NO}}\right)  $\ be
a promise problem in $\mathsf{P{}romiseBQP}^{A}\setminus\mathsf{P{}%
romisePH}^{A}$, for some oracle $A$. \ Then clearly, every $\mathsf{P{}%
romisePH}^{A}$\ machine $M$ fails to solve $\Pi$ on infinitely many inputs $x$
in $\Pi_{\operatorname*{YES}}\cup\Pi_{\operatorname*{NO}}$. \ This means that
we can produce an infinite sequence of inputs $x_{1},x_{2},\ldots$ in
$\Pi_{\operatorname*{YES}}\cup\Pi_{\operatorname*{NO}}$, whose lengths
$n_{1},n_{2},\ldots$ are spaced arbitrarily far apart, such that every
$\mathsf{P{}romisePH}^{A}$\ machine $M$\ fails to solve $\Pi$ on at least one
$x_{i}$. \ Now let $B$\ be an oracle that is identical to $A$, except that for
each input length $n$, it reveals (i) whether $n=n_{i}$ for some $i$ and (ii)
if so, what the corresponding $x_{i}$\ is. \ Also, let $L$\ be the unary
language that contains $0^{n}$\ if and only if (i) $n=n_{i}$\ for some $i$ and
(ii) $x_{i}\in\Pi_{\operatorname*{YES}}$. \ Then $L$\ is in $\mathsf{BQP}^{B}%
$\ but not $\mathsf{PH}^{B}$.}

\subsection{Motivation\label{MOTIVATION}}

There are at least four reasons why the $\mathsf{BQP}$ versus $\mathsf{PH}%
$\ question is so interesting. \ At a basic level, it is both theoretically
and practically important to understand what classical resources are needed to
simulate quantum physics. \ For example, when a quantum system evolves to a
given state, is there always a short classical proof that it does so? \ Can
one estimate quantum amplitudes using approximate counting (which would imply
$\mathsf{BQP}\subseteq\mathsf{BPP}^{\mathsf{NP}}$)? \ If something like this
were true, then while the exponential speedup of Shor's factoring algorithm
might stand, quantum computing would nevertheless seem much less different
from classical computing than previously thought.

Second, if $\mathsf{BQP}\not \subset \mathsf{PH}$, then many possibilities for
new quantum algorithms might open up to us. \ One often hears the complaint
that there are too few quantum algorithms, or that progress on quantum
algorithms has slowed since the mid-1990s. \ In our opinion, the real issue
here has nothing to do with quantum computing, and is simply that there are
too few natural $\mathsf{NP}$-intermediate problems for which there plausibly
\textit{could be} quantum algorithms! \ In other words, instead of focussing
on \textsc{Graph Isomorphism}\ and a small number of other $\mathsf{NP}%
$-intermediate problems, it might be fruitful to look for quantum algorithms
solving completely different types of problems---problems that are not
necessarily even in $\mathsf{PH}$. \ In this paper, we will see a new example
of such a quantum algorithm, which solves a problem called \textsc{Fourier
Checking}.

Third, it is natural to ask whether the $\mathsf{P}\overset{?}{=}\mathsf{BQP}%
$\ question is related to that \textit{other} fundamental question of
complexity theory, $\mathsf{P}\overset{?}{=}\mathsf{NP}$. \ More concretely,
is it possible that quantum computers could provide exponential speedups even
if $\mathsf{P}=\mathsf{NP}$? \ If $\mathsf{BQP}\subseteq\mathsf{PH}$, then
certainly the answer to that question is no (since $\mathsf{P}=\mathsf{NP}%
\Longrightarrow\mathsf{P}=\mathsf{PH}$). \ Therefore, if we want evidence that
quantum computing could survive a collapse of $\mathsf{P}$\ and $\mathsf{NP}$,
we must also seek evidence that $\mathsf{BQP}\not \subset \mathsf{PH}$.

Fourth, a major challenge for quantum computing research is to \textit{get
better evidence that quantum computers cannot solve }$\mathsf{NP}%
$\textit{-complete problems in polynomial time}. \ As an example,\ could we
show that if $\mathsf{NP}\subseteq\mathsf{BQP}$, then the polynomial hierarchy
collapses? \ At first glance, this seems like a wild hope; certainly we have
no idea at present how to prove anything of the kind. \ However, notice that
if $\mathsf{BQP}\subseteq\mathsf{AM}$, then the desired implication would
follow immediately! \ For in that case,%
\begin{align*}
\mathsf{NP}\subseteq\mathsf{BQP} &  \Longrightarrow\mathsf{coNP}%
\subseteq\mathsf{BQP}\\
&  \Longrightarrow\mathsf{coNP}\subseteq\mathsf{AM}\\
&  \Longrightarrow\mathsf{PH}=\mathsf{\Sigma}_{\mathsf{2}}^{\mathsf{P}}%
\end{align*}
where the last implication was shown by Boppana, H\aa stad, and Zachos
\cite{bhz}. \ Similar remarks apply to the questions of whether $\mathsf{NP}%
\subseteq\mathsf{BQP}$\ would imply $\mathsf{PH}\subseteq\mathsf{BQP}$, and
whether the folklore result $\mathsf{NP}^{\mathsf{BPP}}\subseteq
\mathsf{BPP}^{\mathsf{NP}}$\ has the quantum analogue $\mathsf{NP}%
^{\mathsf{BQP}}\subseteq\mathsf{BQP}^{\mathsf{NP}}$. \ In each of these cases,
we find that understanding some other issue in quantum complexity theory
requires first coming to grips with whether $\mathsf{BQP}$\ is contained in
some level of the polynomial hierarchy.

\subsection{Our Results\label{RESULTS}}

This paper presents the first formal evidence for the possibility that
$\mathsf{BQP}\not \subset \mathsf{PH}$. \ Perhaps more importantly, it places
the relativized $\mathsf{BQP}$ versus $\mathsf{PH}$\ question at the frontier
of (classical) circuit lower bounds. The heart of the problem, we will find,
is to extend Braverman's spectacular recent proof \cite{braverman}\ of the
Linial-Nisan Conjecture, in ways that would reveal a great deal of information
about small-depth circuits independent of the implications for quantum computing.

We have two main contributions. \ First, we achieve an oracle separation
between $\mathsf{BQP}$\ and $\mathsf{PH}$ for the case of \textit{relation
problems}. \ A relation problem is simply a problem where the desired output
is an $n$-bit string (rather than a single bit), and any string from some
nonempty set $S$ is acceptable. \ Relation problems arise often in theoretical
computer science; one well-known example is finding a Nash equilibrium (shown
to be $\mathsf{PPAD}$-complete by Daskalakis et al.\ \cite{dgp}). \ Within
quantum computing, there is considerable precedent for studying relation
problems as a warmup to the harder case of decision problems. \ For example,
in 2004 Bar-Yossef, Jayram, and Kerenidis \cite{bjk} gave a relation problem
with quantum one-way communication complexity $O\left(  \log n\right)  $ and
randomized one-way communication complexity $\Omega\left(  \sqrt{n}\right)  $.
\ It took several more years for Gavinsky et al.\ \cite{gkkrw} to achieve the
same separation for decision problems, and the proof was much more
complicated. \ The same phenomenon has arisen many times in quantum
communication complexity
\cite{gavinsky:ci,gavinsky:se,gavinsky:pl,gkrw,gavinskypudlak}, though to our
knowledge, this is the first time it has arisen in quantum query complexity.

Formally, our result is as follows:

\begin{theorem}
\label{thm1}There exists an oracle $A$ relative to which $\mathsf{FBQP}%
^{A}\not \subset \mathsf{FBPP}^{\mathsf{PH}^{A}}$, where $\mathsf{FBQP}$\ and
$\mathsf{FBPP}$\ are the relation versions of $\mathsf{BQP}$\ and
$\mathsf{BPP}$\ respectively.\footnote{Confusingly, the $\mathsf{F}$\ stands
for \textquotedblleft function\textquotedblright; we are simply following the
standard naming convention for classes of relation problems ($\mathsf{FP}$,
$\mathsf{FNP}$, etc).}
\end{theorem}

Underlying Theorem \ref{thm1} is a new lower bound against $\mathsf{AC}^{0}%
$\ circuits (constant-depth circuits composed of AND, OR, and NOT gates).
\ The close connection between $\mathsf{AC}^{0}$\ and the polynomial hierarchy
that we exploit is not new. \ In the early 1980s, Furst-Saxe-Sipser \cite{fss}
and Yao \cite{yao:ph}\ noticed that, if we have a $\mathsf{PH}$\ machine
$M$\ that computes (say) the \textsc{Parity} of a $2^{n}$-bit oracle string,
then by simply reinterpreting the existential quantifiers of $M$ as OR gates
and the universal quantifiers as AND gates, we obtain an $\mathsf{AC}^{0}%
$\ circuit of size $2^{\operatorname*{poly}\left(  n\right)  }$ solving the
same problem. \ It follows that, if we can prove a $2^{\omega\left(
\operatorname*{polylog}n\right)  }$\ lower bound on the size of $\mathsf{AC}%
^{0}$\ circuits computing \textsc{Parity}, we can construct an oracle $A$
relative to which $\mathsf{\oplus P}^{A}\not \subset \mathsf{PH}^{A}$. \ The
idea is the same for constructing an $A$ relative to which $\mathcal{C}%
^{A}\not \subset \mathsf{PH}^{A}$, where $\mathcal{C}$\ is any complexity class.

Indeed, the relation between $\mathsf{PH}$\ and $\mathsf{AC}^{0}$\ is so
direct that we get the following as a more-or-less immediate counterpart to
Theorem \ref{thm1}:

\begin{theorem}
\label{bqlogac0}In the unrelativized world (with no oracle), there exists a
relation problem solvable in quantum logarithmic time but not in nonuniform
$\mathsf{AC}^{0}$.
\end{theorem}

The relation problem that we use to separate $\mathsf{BQP}$\ from
$\mathsf{PH}$, and $\mathsf{BQLOGTIME}$\ from $\mathsf{AC}^{0}$, is called
\textsc{Fourier Fishing}. \ The problem can be informally stated as follows.
\ We are given oracle access to $n$ Boolean functions $f_{1},\ldots
,f_{n}:\left\{  0,1\right\}  ^{n}\rightarrow\left\{  -1,1\right\}  $, which we
think of as chosen uniformly at random. \ The task is to output $n$ strings,
$z_{1},\ldots,z_{n}\in\left\{  0,1\right\}  ^{n}$, such that the corresponding
squared Fourier coefficients $\widehat{f}_{1}\left(  z_{1}\right)  ^{2}%
,\ldots,\widehat{f}_{n}\left(  z_{n}\right)  ^{2}$\ are \textquotedblleft
often much larger than average.\textquotedblright\ \ Notice that if $f_{i}$ is
a random Boolean function, then each of its Fourier coefficients $\widehat
{f}_{i}\left(  z\right)  $ follows a normal distribution---meaning that with
overwhelming probability, a constant fraction of the Fourier coefficients will
be a constant factor larger than the mean. \ Furthermore, it is
straightforward to create a quantum algorithm that samples each $z$ with
probability proportional to $\widehat{f}_{i}\left(  z\right)  ^{2}$, so that
larger Fourier coefficients are more likely to be sampled than smaller ones.

On the other hand, computing any \textit{specific} $\widehat{f}_{i}\left(
z\right)  $\ is easily seen to be equivalent to summing $2^{n}$\ bits. \ By
well-known lower bounds on the size of $\mathsf{AC}^{0}$\ circuits computing
the \textsc{Majority} function (see H\aa stad \cite{hastad:book}\ for
example), it follows that, for any fixed $z$, computing $\widehat{f}%
_{i}\left(  z\right)  $\ cannot be in $\mathsf{PH}$ as an oracle problem.
\ Unfortunately, this does not directly imply any separation between
$\mathsf{BQP}$\ and $\mathsf{PH}$, since the quantum algorithm does not
compute $\widehat{f}_{i}\left(  z\right)  $\ either: it just samples a $z$
with probability proportional to $\widehat{f}_{i}\left(  z\right)  ^{2}$.
\ However, we will show that, if there exists a $\mathsf{BPP}^{\mathsf{PH}}%
$\ machine $M$\ that even \textit{approximately} simulates the behavior of the
quantum algorithm, then one can solve \textsc{Majority}\ by means of a
\textit{nondeterministic} reduction---which uses approximate counting to
estimate $\Pr\left[  M\text{ outputs }z\right]  $, and adds a constant number
of layers to the $\mathsf{AC}^{0}$\ circuit. \ The central difficulty is that,
if $M$\ knew the specific $z$ for which we were interested in estimating
$\widehat{f}_{i}\left(  z\right)  $, then it could choose adversarially never
to output that $z$. \ To solve this, we will show that we can
\textquotedblleft smuggle\textquotedblright\ a \textsc{Majority}\ instance
into the estimation of a \textit{random} Fourier coefficient $\widehat{f}%
_{i}\left(  z\right)  $, in such a way that it is information-theoretically
impossible for $M$ to determine which $z$ we care about.

Our second contribution is to define and study a new black-box decision
problem, called \textsc{Fourier Checking}. \ Informally, in this problem we
are given oracle access to \textit{two} Boolean functions $f,g:\left\{
0,1\right\}  ^{n}\rightarrow\left\{  -1,1\right\}  $, and are promised that either

\begin{enumerate}
\item[(i)] $f$ and $g$ are both uniformly random, or

\item[(ii)] $f$ is uniformly random, while $g$ is extremely well correlated
with $f$'s Fourier transform over $\mathbb{Z}_{2}^{n}$ (which we call
\textquotedblleft forrelated\textquotedblright).
\end{enumerate}

\noindent The problem is to decide whether (i) or (ii) is the case.

It is not hard to show that \textsc{Fourier Checking} is in $\mathsf{BQP}$:
basically, one can prepare a uniform superposition over all $x\in\left\{
0,1\right\}  ^{n}$, then query $f$, apply a quantum Fourier transform, query
$g$, and check whether one has recovered something close to the uniform
superposition. \ On the other hand, being forrelated seems like an extremely
\textquotedblleft global\textquotedblright\ property of $f$ and $g$: one that
would not be apparent from querying \textit{any} small number of $f\left(
x\right)  $\ and $g\left(  y\right)  $\ values, regardless of the outcomes of
those queries. \ And thus, one might conjecture that \textsc{Fourier Checking}
(as an oracle problem) is not in $\mathsf{PH}$.

In this paper, we adduce strong evidence for that conjecture. \ Specifically,
we show that for every $k\leq2^{n/4}$, the forrelated distribution over
$\left\langle f,g\right\rangle $\ pairs is $O\left(  k^{2}/2^{n/2}\right)
$\textit{-almost} $k$\textit{-wise independent}. \ By this we mean that, if
one had $1/2$\ prior probability that $f$ and $g$ were uniformly random, and
$1/2$\ prior probability that $f$ and $g$ were forrelated, then even
conditioned on any $k$ values of $f$ and $g$, the posterior probability that
$f$ and $g$ were forrelated would still be%
\[
\frac{1}{2}\pm O\left(  \frac{k^{2}}{2^{n/2}}\right)  .
\]
We conjecture that this almost $k$-wise independence property is enough, by
itself, to imply that an oracle problem is not in $\mathsf{PH}$. \ We call
this the \textit{Generalized Linial-Nisan Conjecture}.

Without the $\pm O\left(  k^{2}/2^{n/2}\right)  $\ error term, our conjecture
would be equivalent\footnote{Up to unimportant variations in the parameters}
to a famous conjecture in circuit complexity made by Linial and\ Nisan
\cite{ln}\ in 1990. \ Their conjecture stated that \textit{polylogarithmic
independence fools }$\mathsf{AC}^{0}$: in other words, every probability
distribution over $N$-bit strings that is \textit{uniform} on every small
subset of bits, is indistinguishable from the truly uniform distribution by
$\mathsf{AC}^{0}$ circuits. \ When we began investigating this topic a year
ago, even the original Linial-Nisan Conjecture was still open. \ Since then,
Braverman \cite{braverman} (building on earlier work by Bazzi \cite{bazzi} and
Razborov \cite{razborov:bazzi}) has given a beautiful proof of that
conjecture. \ In other words, to construct an oracle relative to which
$\mathsf{BQP}\not \subset \mathsf{PH}$, it now suffices to generalize
Braverman's Theorem from $k$-wise independent distributions to almost $k$-wise
independent ones. \ We believe that this is by far the most promising approach
to the $\mathsf{BQP}$\ versus $\mathsf{PH}$\ problem.

Alas, generalizing Braverman's proof is much harder than one might have hoped.
\ To prove the original Linial-Nisan Conjecture, Braverman showed that every
$\mathsf{AC}^{0}$\ function\ $f:\left\{  0,1\right\}  ^{n}\rightarrow\left\{
0,1\right\}  $\ can be well-approximated, in the $L_{1}$-norm, by
\textit{low-degree sandwiching polynomials}: real polynomials $p_{\ell}%
,p_{u}:\mathbb{R}^{n}\rightarrow\mathbb{R}$, of degree $O\left(
\operatorname*{polylog}n\right)  $, such that $p_{\ell}\left(  x\right)  \leq
f\left(  x\right)  \leq p_{u}\left(  x\right)  $\ for all $x\in\left\{
0,1\right\}  ^{n}$. \ Since $p_{\ell}$\ and $p_{u}$ trivially have the same
expectation on any $k$-wise independent distribution that they have on the
uniform distribution, one can show that $f$ must have almost the same
expectation as well. \ To generalize Braverman's result from $k$-wise
independence to almost $k$-wise independence, we will show that it suffices to
construct low-degree sandwich polynomials that satisfy a certain additional
condition. \ This new condition (which we call \textquotedblleft
low-fat\textquotedblright) basically says that $p_{\ell}$\ and $p_{u}$\ must
be representable as linear combinations of \textit{terms} (that is, products
of $x_{i}$'s and $\left(  1-x_{i}\right)  $'s), in such a way that the sum of
the absolute values of the coefficients is bounded---thereby preventing
\textquotedblleft massive cancellations\textquotedblright\ between positive
and negative terms. \ Unfortunately, while we know two techniques for
approximating $\mathsf{AC}^{0}$\ functions by low-degree polynomials---that of
Linial-Mansour-Nisan \cite{lmn} and that of Razborov \cite{razborov:ac0}\ and
Smolensky \cite{smolensky}---neither technique provides anything like the
control over coefficients that we need. \ To construct low-fat sandwiching
polynomials, it seems necessary to reprove the LMN and Razborov-Smolensky
theorems in a more \textquotedblleft conservative,\textquotedblright\ less
\textquotedblleft profligate\textquotedblright\ way. \ And such an advance
seems likely to lead to breakthroughs in circuit complexity and computational
learning theory having nothing to do with quantum computing.

Let us mention two further applications of \textsc{Fourier Checking}:

\begin{enumerate}
\item[(1)] If the Generalized Linial-Nisan Conjecture holds, then just like
with \textsc{Fourier Fishing}, we can \textquotedblleft scale down by an
exponential,\textquotedblright\ to obtain a promise problem that is in
$\mathsf{BQLOGTIME}$\ but not in $\mathsf{AC}^{0}$.

\item[(2)] Without any assumptions, we can prove the new results that there
exist oracles relative to which $\mathsf{BQP}\not \subset \mathsf{BPP}%
_{\mathsf{path}}$ and $\mathsf{BQP}\not \subset \mathsf{SZK}$. \ We can also
reprove all previous oracle separations between $\mathsf{BQP}$\ and classical
complexity classes in a unified fashion.
\end{enumerate}

To summarize our conclusions:

\begin{theorem}
\label{thm2}Assuming the Generalized Linial-Nisan Conjecture, there exists an
oracle $A$ relative to which $\mathsf{BQP}^{A}\not \subset \mathsf{PH}^{A}$,
and there also exists a promise problem in $\mathsf{BQLOGTIME}\setminus
\mathsf{AC}^{0}$. \ Unconditionally, there exists an oracle $A$ relative to
which $\mathsf{BQP}^{A}\not \subset \mathsf{BPP}_{\mathsf{path}}^{A}$ and
$\mathsf{BQP}^{A}\not \subset \mathsf{SZK}^{A}$.
\end{theorem}

As a candidate problem, \textsc{Fourier Checking} has at least five advantages
over the \textsc{Recursive Fourier Sampling}\ problem of Bernstein and
Vazirani \cite{bv}. \ First, it is much simpler to define and reason about.
\ Second, \textsc{Fourier Checking} has the almost $k$-wise
independence\ property, which is not shared by \textsc{Recursive Fourier
Sampling}, and which immediately connects the former to general questions
about pseudorandomness against constant-depth circuits. \ Third,
\textsc{Fourier Checking} can yield exponential separations between quantum
and classical models, rather than just quasipolynomial ones. \ Fourth, one can
hope to use \textsc{Fourier Checking}\ to give an oracle relative to which
$\mathsf{BQP}$\ is not in $\mathsf{PH}\left[  n^{c}\right]  $ (or
$\mathsf{PH}$\ with $n^{c}$\ alternations) for any fixed $c$;\ by contrast,
\textsc{Recursive Fourier Sampling} is in $\mathsf{PH}\left[  \log n\right]
$. \ Finally, it is at least conceivable that the quantum algorithm for
\textsc{Fourier Checking}\ is \textit{good} for something. \ We leave the
challenge of finding an explicit computational problem that \textquotedblleft
instantiates\textquotedblright\ \textsc{Fourier Checking}, in the same way
that \textsc{Factoring}\ and \textsc{Discrete Logarithm}\ instantiated Shor's
period-finding problem.

\subsection{In Defense of Oracles\label{ORACLEDEF}}

This paper is concerned with finding oracles relative to which $\mathsf{BQP}%
$\ outperforms classical complexity classes. \ As such, it is open to the
usual objections: \textquotedblleft But don't oracle results mislead us about
the `real' world? \ What about non-relativizing results like $\mathsf{IP}%
=\mathsf{PSPACE}$ \cite{shamir}?\textquotedblright

In our view, it is most helpful to think of oracle separations, not as strange
metamathematical claims, but as \textit{lower bounds in a concrete
computational model that is natural and well-motivated in its own right.}
\ The model in question is \textit{query complexity}, where the resource to be
minimized is the number of accesses to a very long input string. \ When
someone gives an oracle $A$ relative to which $\mathcal{C}^{A}\not \subset
\mathcal{D}^{A}$, what they really mean is simply that they have found a
problem that $\mathcal{C}$ machines can solve using superpolynomially fewer
queries than $\mathcal{D}$\ machines. \ In other words, $\mathcal{C}$\ has has
\textquotedblleft cleared the first possible obstacle\textquotedblright---the
query complexity obstacle---to having capabilities beyond those of
$\mathcal{D}$. \ Of course, it could be (and sometimes is) that $\mathcal{C}%
\subseteq\mathcal{D}$\ for other reasons,\ but if we do not \textit{even} have
a query complexity lower bound, then proving one is in some sense the obvious
place to start.

Oracle separations have played a role in many of the central developments of
both classical and quantum complexity theory. \ As mentioned earlier, proving
query complexity lower bounds for $\mathsf{PH}$\ machines is essentially
equivalent to proving size lower bounds for $\mathsf{AC}^{0}$\ circuits---and
indeed, the pioneering $\mathsf{AC}^{0}$\ lower bounds of the early 1980s were
explicitly motivated by the goal of proving oracle separations for
$\mathsf{PH}$.\footnote{Yao's paper \cite{yao:ph} was entitled
\textquotedblleft Separating the polynomial-time hierarchy by
oracles\textquotedblright; the Furst-Saxe-Sipser paper \cite{fss}\ was
entitled \textquotedblleft Parity, circuits, and the polynomial time
hierarchy.\textquotedblright} \ Within quantum computing, oracle results have
played an even more decisive role: the first evidence for the power of quantum
computers came from the oracle separations of Bernstein-Vazirani \cite{bv} and
Simon \cite{simon}, and Shor's algorithm \cite{shor} contains an oracle
algorithm (for the \textsc{Period-Finding}\ problem) at its core.

Having said all that, if for some reason one still feels averse to the
language of oracles, then (as mentioned before) one is free to scale
everything down by an exponential,\ and to reinterpret a relativized
separation between $\mathsf{BQP}$\ and $\mathsf{PH}$\ as an \textit{un}%
relativized separation between $\mathsf{BQLOGTIME}$\ and $\mathsf{AC}^{0}$.

\section{Preliminaries\label{PRELIM}}

It will be convenient to consider Boolean functions of the form $f:\left\{
0,1\right\}  ^{n}\rightarrow\left\{  -1,1\right\}  $. \ Throughout this paper,
we let $N=2^{n}$;\ we will often view the truth table of a Boolean function as
an \textquotedblleft input\textquotedblright\ of size $N$. \ Given a Boolean
function $f:\left\{  0,1\right\}  ^{n}\rightarrow\left\{  -1,1\right\}  $, the
Fourier transform of $f$ is defined as%
\[
\widehat{f}\left(  z\right)  :=\frac{1}{\sqrt{N}}\sum_{x\in\left\{
0,1\right\}  ^{n}}\left(  -1\right)  ^{x\cdot z}f\left(  x\right)  .
\]
Recall Parseval's identity:%
\[
\sum_{x\in\left\{  0,1\right\}  ^{n}}f\left(  x\right)  ^{2}=\sum
_{z\in\left\{  0,1\right\}  ^{n}}\widehat{f}\left(  z\right)  ^{2}=N.
\]

\subsection{Problems\label{PROBLEMS}}

We first define the \textsc{Fourier Fishing} problem, in both
\textquotedblleft distributional\textquotedblright\ and \textquotedblleft
promise\textquotedblright\ versions. \ In the distributional version, we are
given oracle access to $n$ Boolean functions $f_{1},\ldots,f_{n}:\left\{
0,1\right\}  ^{n}\rightarrow\left\{  -1,1\right\}  $, which are chosen
uniformly and independently at random. \ The task is to output $n$ strings,
$z_{1},\ldots,z_{n}\in\left\{  0,1\right\}  ^{n}$, at least 75\% of which
satisfy $\left\vert \widehat{f}_{i}\left(  z_{i}\right)  \right\vert \geq1$
and at least 25\% of which satisfy $\left\vert \widehat{f}_{i}\left(
z_{i}\right)  \right\vert \geq2$. \ (Note that these thresholds are not
arbitrary, but were carefully chosen to produce a separation between the
quantum and classical models!)

We now want a version of \textsc{Fourier Fishing}\ that removes the need to
assume the $f_{i}$'s are uniformly random, replacing it with a worst-case
promise on the $f_{i}$'s. \ Call an $n$-tuple $\left\langle f_{1},\ldots
,f_{n}\right\rangle $\ of Boolean functions\ \textit{good} if%
\begin{align*}
\sum_{i=1}^{n}\sum_{z_{i}:\left\vert \widehat{f}_{i}\left(  z_{i}\right)
\right\vert \geq1}\widehat{f}_{i}\left(  z_{i}\right)  ^{2}  &  \geq0.8Nn,\\
\sum_{i=1}^{n}\sum_{z_{i}:\left\vert \widehat{f}_{i}\left(  z_{i}\right)
\right\vert \geq2}\widehat{f}_{i}\left(  z_{i}\right)  ^{2}  &  \geq0.26Nn.
\end{align*}
(We will show in Lemma \ref{mostgood}\ that the vast majority of $\left\langle
f_{1},\ldots,f_{n}\right\rangle $ are good.) \ In \textsc{Promise Fourier
Fishing}, we are given oracle access to Boolean functions $f_{1},\ldots
,f_{n}:\left\{  0,1\right\}  ^{n}\rightarrow\left\{  -1,1\right\}  $, which
are promised to be good.\ \ The task, again, is to output strings
$z_{1},\ldots,z_{n}\in\left\{  0,1\right\}  ^{n}$, at least 75\% of which
satisfy $\left\vert \widehat{f}_{i}\left(  z_{i}\right)  \right\vert \geq1$
and at least 25\% of which satisfy $\left\vert \widehat{f}_{i}\left(
z_{i}\right)  \right\vert \geq2$.

Next we define a decision problem called \textsc{Fourier Checking}. \ Here we
are given oracle access to two Boolean functions $f,g:\left\{  0,1\right\}
^{n}\rightarrow\left\{  -1,1\right\}  $. \ We are promised that either

\begin{enumerate}
\item[(i)] $\left\langle f,g\right\rangle $ was drawn from the uniform
distribution $\mathcal{U}$, which sets every $f\left(  x\right)  $\ and
$g\left(  y\right)  $\ by a fair, independent coin toss.

\item[(ii)] $\left\langle f,g\right\rangle $ was drawn from the
\textquotedblleft forrelated\textquotedblright\ distribution\ $\mathcal{F}$,
which is defined as follows. \ First choose a random real vector $v=\left(
v_{x}\right)  _{x\in\left\{  0,1\right\}  ^{n}}\in\mathbb{R}^{N}$, by drawing
each entry independently from a Gaussian distribution with mean $0$ and
variance $1$. \ Then set $f\left(  x\right)  :=\operatorname*{sgn}\left(
v_{x}\right)  $\ and $g\left(  x\right)  :=\operatorname*{sgn}\left(
\widehat{v}_{x}\right)  $\ for all $x$. \ Here%
\[
\operatorname*{sgn}\left(  \alpha\right)  :=\left\{
\begin{array}
[c]{cc}%
1 & \text{if }\alpha\geq0\\
-1 & \text{if }\alpha<0
\end{array}
\right.
\]
and $\widehat{v}$ is the Fourier transform of $v$\ over $\mathbb{Z}_{2}^{n}$:%
\[
\widehat{v}_{y}:=\frac{1}{\sqrt{N}}\sum_{x\in\left\{  0,1\right\}  ^{n}%
}\left(  -1\right)  ^{x\cdot y}v_{x}.
\]
In other words, $f$\ and $g$\ \textit{individually}\ are still uniformly
random, but they are no longer independent: now $g$ is now extremely well
correlated with the Fourier transform of $f$ (hence \textquotedblleft
forrelated\textquotedblright).
\end{enumerate}

\noindent The problem is to accept if $\left\langle f,g\right\rangle $ was
drawn from $\mathcal{F}$, and to reject if $\left\langle f,g\right\rangle $
was drawn from $\mathcal{U}$. \ Note that, since $\mathcal{F}$\ and
$\mathcal{U}$\ overlap slightly, we can only hope to succeed with overwhelming
probability over the choice of $\left\langle f,g\right\rangle $,\ not for
every $\left\langle f,g\right\rangle $ pair.

We can also define a promise-problem version of \textsc{Fourier Checking}.
\ In \textsc{Promise Fourier Checking}, we are promised that the quantity%
\[
p\left(  f,g\right)  :=\frac{1}{N^{3}}\left(  \sum_{x,y\in\left\{
0,1\right\}  ^{n}}f\left(  x\right)  \left(  -1\right)  ^{x\cdot y}g\left(
y\right)  \right)  ^{2}%
\]
is either at least $0.05$\ or at most $0.01$. \ The problem is to accept in
the former case and reject in the latter case.

\subsection{Complexity Classes\label{CC}}

See the Complexity Zoo\footnote{www.complexityzoo.com} for the definitions of
standard complexity classes, such as $\mathsf{BQP}$, $\mathsf{AM}$, and
$\mathsf{PH}$. \ When we write $\mathcal{C}^{\mathsf{PH}}$\ (i.e., a
complexity class $\mathcal{C}$\ with an oracle for the polynomial hierarchy),
we mean $\cup_{k\geq1}\mathcal{C}^{\mathsf{\Sigma}_{k}^{\mathsf{P}}}$.

We will consider not only decision problems, but also \textit{relation
problems} (also called \textit{function problems}). \ In a relation problem,
the output is not a single bit but a $\operatorname*{poly}\left(  n\right)
$-bit string $y$. \ There could be many valid $y$'s for a given instance, and
the algorithm's task is to output any one of them.

The definitions of $\mathsf{FP}$\ and $\mathsf{FNP}$\ (the relation versions
of $\mathsf{P}$\ and $\mathsf{NP}$) are standard. \ We now define
$\mathsf{FBPP}$\ and $\mathsf{FBQP}$, the relation versions of $\mathsf{BPP}%
$\ and $\mathsf{BQP}$.

\begin{definition}
$\mathsf{FBPP}$ is the class of relations $R\subseteq\left\{  0,1\right\}
^{\ast}\times\left\{  0,1\right\}  ^{\ast}$\ for which there exists a
probabilistic polynomial-time algorithm $A$ that, given any input
$x\in\left\{  0,1\right\}  ^{n}$, produces an output $y$ such that%
\[
\Pr\left[  \left(  x,y\right)  \in R\right]  =1-o\left(  1\right)  ,
\]
where the probability is over $A$'s internal randomness. \ (In particular,
this implies that for every $x$, there exists at least one $y$ such that
$\left(  x,y\right)  \in R$.) \ $\mathsf{FBQP}$\ is defined the same way,
except that $A$ is a quantum algorithm rather than a classical one.
\end{definition}

An important point about $\mathsf{FBPP}$\ and $\mathsf{FBQP}$\ is that, as far
as we know, these classes\ do not admit amplification. \ In other words, the
value of an algorithm's success probability might actually matter, not just
the fact that the probability is bounded above $1/2$. \ This is why we adopt
the convention that an algorithm \textquotedblleft succeeds\textquotedblright%
\ if it outputs $\left(  x,y\right)  \in R$\ with probability $1-o\left(
1\right)  $. \ In practice, we will give oracle problems for which the
$\mathsf{FBQP}$\ algorithm succeeds with probability $1-1/\exp\left(
n\right)  $, while any $\mathsf{FBPP}^{\mathsf{PH}}$\ algorithm succeeds with
probability at most (say) $0.99$. \ How far the constant in this separation
can be improved is an open problem.

Another important point is that, while $\mathsf{BPP}^{\mathsf{PH}}%
=\mathsf{P}^{\mathsf{PH}}$\ (which follows from $\mathsf{BPP}\subseteq
\mathsf{\Sigma}_{\mathsf{2}}^{\mathsf{P}}$), the class $\mathsf{FBPP}%
^{\mathsf{PH}}$\ is strictly larger than $\mathsf{FP}^{\mathsf{PH}}$. \ To see
this, consider the relation%
\[
R=\left\{  \left(  0^{n},y\right)  :K\left(  y\right)  \geq n\right\}  ,
\]
where we are given $n$, and asked to output any string of Kolmogorov
complexity at least $n$. \ Clearly this problem is in $\mathsf{FBPP}$: just
output a random $2n$-bit string. \ On the other hand, just as obviously the
problem is not in $\mathsf{FP}^{\mathsf{PH}}$. \ This is why we need to
construct an oracle $A$ such that $\mathsf{FBQP}^{A}\not \subset
\mathsf{FBPP}^{\mathsf{PH}^{A}}$: because constructing an oracle $A$ such that
$\mathsf{FBQP}^{A}\not \subset \mathsf{FP}^{\mathsf{PH}^{A}}$\ is trivial and
not even related to quantum computing.

We now discuss some \textquotedblleft low-level\textquotedblright\ complexity
classes. \ $\mathsf{AC}^{0}$\ is the class of problems solvable by a
nonuniform family of AND/OR/NOT circuits, with depth $O\left(  1\right)  $,
size $\operatorname*{poly}\left(  n\right)  $, and unbounded fanin. \ When we
say \textquotedblleft$\mathsf{AC}^{0}$ circuit,\textquotedblright\ we mean a
constant-depth circuit of AND/OR/NOT gates, not necessarily of polynomial
size. \ Any such circuit can be made into a \textit{formula} (i.e., a circuit
of fanout $1$) with only a polynomial increase in size. \ The circuit has
\textit{depth} $d$ if it consists of $d$ alternating layers of AND and OR
gates (without loss of generality, the NOT gates can all be pushed to the
bottom, and we do not count them towards the depth). \ For example, a DNF
(Disjunctive Normal Form) formula is just an $\mathsf{AC}^{0}$\ circuit of
depth $2$.

We will also be interested in quantum \textit{logarithmic} time, which can be
defined naturally as follows:

\begin{definition}
$\mathsf{BQLOGTIME}$ is the class of languages $L\subseteq\left\{
0,1\right\}  ^{\ast}$\ that are decidable, with bounded probability of error,
by a $\mathsf{LOGTIME}$-uniform family of quantum circuits $\left\{
C_{n}\right\}  _{n}$\ such that each $C_{n}$ has $O\left(  \log n\right)
$\ gates, and can include gates that make random-access queries to the input
string $x=x_{1}\ldots x_{n}$ (i.e., that map $\left\vert i\right\rangle
\left\vert z\right\rangle $ to $\left\vert i\right\rangle \left\vert z\oplus
x_{i}\right\rangle $ for every $i\in\left[  n\right]  $).
\end{definition}

One other complexity class that arises in this paper, which is less well known
than it should be, is $\mathsf{BPP}_{\mathsf{path}}$. \ Loosely speaking,
$\mathsf{BPP}_{\mathsf{path}}$ can be defined as the class of problems that
are solvable in probabilistic polynomial time, given the ability to
\textquotedblleft postselect\textquotedblright\ (that is, discard all runs of
the computation that do not produce a desired result, even if such runs are
the overwhelming majority). \ Formally:

\begin{definition}
$\mathsf{BPP}_{\mathsf{path}}$ is the class of languages $L\subseteq\left\{
0,1\right\}  ^{\ast}$\ for which there exists a $\mathsf{BPP}$\ machine $M$,
which can either \textquotedblleft succeed\textquotedblright\ or
\textquotedblleft fail\textquotedblright\ and conditioned on succeeding either
\textquotedblleft accept\textquotedblright\ or \textquotedblleft
reject,\textquotedblright\ such that for all inputs $x$:

\begin{enumerate}
\item[(i)] $\Pr\left[  M\left(  x\right)  \text{ succeeds}\right]  >0$.

\item[(ii)] $x\in L\Longrightarrow\Pr\left[  M\left(  x\right)  \text{ accepts
}|~M\left(  x\right)  \text{ succeeds}\right]  \geq\frac{2}{3}$.

\item[(iii)] $x\notin L\Longrightarrow\Pr\left[  M\left(  x\right)  \text{
accepts }|~M\left(  x\right)  \text{ succeeds}\right]  \leq\frac{1}{3}$.
\end{enumerate}
\end{definition}

$\mathsf{BPP}_{\mathsf{path}}$ was defined by Han, Hemaspaandra, and Thierauf
\cite{hht}, who also showed that $\mathsf{MA}\subseteq\mathsf{BPP}%
_{\mathsf{path}}$\ and $\mathsf{P}_{||}^{\mathsf{NP}}\subseteq\mathsf{BPP}%
_{\mathsf{path}}\subseteq\mathsf{BPP}_{||}^{\mathsf{NP}}$. \ Using
\textsc{Fourier Checking}, we will construct an oracle $A$ relative to which
$\mathsf{BQP}^{A}\not \subset \mathsf{BPP}_{\mathsf{path}}^{A}$. \ This result
might not sound amazing, but (i) it is new, (ii) it does not follow from the
\textquotedblleft standard\textquotedblright\ quantum algorithms, such as
those of Simon \cite{simon}\ and Shor \cite{shor}, and (iii) it supersedes
almost all previous oracle results placing $\mathsf{BQP}$\ outside classical
complexity classes.\footnote{The one exception is the result of Green and
Pruim \cite{gp}\ that there exists an $A$ relative to which $\mathsf{BQP}%
^{A}\not \subset \mathsf{P}^{\mathsf{NP}^{A}}$,\ but that can also be easily
reproduced using \textsc{Fourier Checking}.} \ As another illustration of the
versatility of \textsc{Fourier Checking}, we use it to give an $A$ such that
$\mathsf{BQP}^{A}\not \subset \mathsf{SZK}^{A}$, where $\mathsf{SZK}$\ is
Statistical Zero Knowledge. \ The opposite direction---an $A$ such that
$\mathsf{SZK}^{A}\not \subset \mathsf{BQP}^{A}$---was shown by Aaronson
\cite{aar:col} in 2002.

\section{Quantum Algorithms\label{QALG}}

In this section, we show that \textsc{Fourier Fishing} and \textsc{Fourier
Checking} both admit simple quantum algorithms.

\subsection{\label{FFALG}Quantum Algorithm for \textsc{Fourier Fishing}}

Here is a quantum algorithm, \texttt{FF-ALG}, that solves \textsc{Fourier
Fishing} with overwhelming probability in $O\left(  n^{2}\right)  $\ time and
$n$\ quantum queries (one to each $f_{i}$). \ For $i:=1$ to $n$, first prepare
the state%
\[
\frac{1}{\sqrt{N}}\sum_{x\in\left\{  0,1\right\}  ^{n}}f_{i}\left(  x\right)
\left\vert x\right\rangle ,
\]
then apply Hadamard gates to all $n$ qubits, then measure in the computational
basis and output the result as $z_{i}$.%
%TCIMACRO{\FRAME{ftbpFU}{4.8084in}{2.4794in}{0pt}{\Qcb{The Fourier coefficients
%of a random Boolean function follow a Gaussian distribution, with mean $0$ and
%variance $1$. \ However, larger Fourier coefficients are more likely to be
%observed by the quantum algorithm.}}{\Qlb{fourierfig}}{fourierfig.eps}%
%width 4.8084in;  height 2.4794in;  depth 0pt;  original-width 10.6718in;
%original-height 8.163in;  cropleft "0.2745";  croptop "0.9625";
%cropright "0.7254";  cropbottom "0.6589";
%filename 'fourierfig.eps';file-properties "XNPEU";}}}%
%BeginExpansion
\begin{figure}
[ptb]
\begin{center}
\includegraphics[
trim=2.929409in 5.378601in 2.930476in 0.306113in,
height=2.4794in,
width=4.8084in
]%
{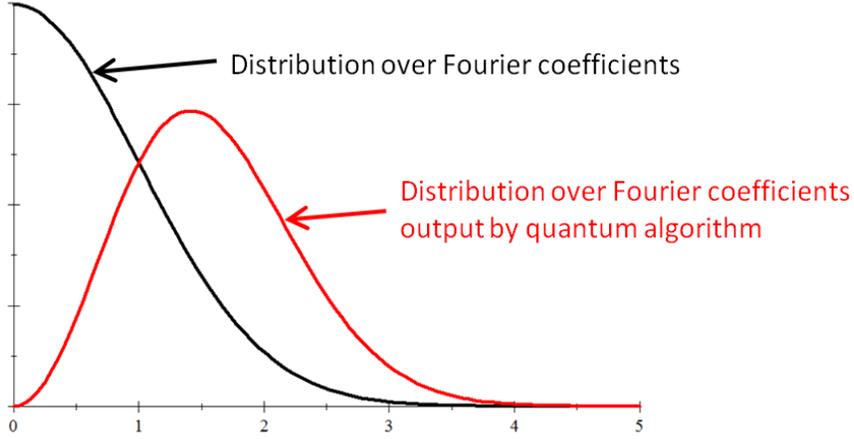}%
\caption{The Fourier coefficients of a random Boolean function follow a
Gaussian distribution, with mean $0$ and variance $1$. \ However, larger
Fourier coefficients are more likely to be observed by the quantum algorithm.}%
\label{fourierfig}%
\end{center}
\end{figure}
%EndExpansion

Intuitively, \texttt{FF-ALG} samples the Fourier coefficients of each $f_{i}$
under a distribution that is skewed towards larger coefficients;\ the
algorithm's behavior is illustrated pictorially in Figure \ref{fourierfig}.
\ We now give a formal analysis. \ Recall the definition of a
\textquotedblleft good\textquotedblright\ tuple $\left\langle f_{1}%
,\ldots,f_{n}\right\rangle $\ from Section \ref{PROBLEMS}. \ Assuming
$\left\langle f_{1},\ldots,f_{n}\right\rangle $\ is good, it is easy to
analyze \texttt{FF-ALG}'s success probability.

\begin{lemma}
\label{ffinbqp}Assuming $\left\langle f_{1},\ldots,f_{n}\right\rangle $\ is
good, \texttt{FF-ALG} succeeds with probability $1-1/\exp\left(  n\right)  $.
\end{lemma}

\begin{proof}
Let $\left\langle z_{1},\ldots,z_{n}\right\rangle $\ be the algorithm's
output. \ For each $i$, let $X_{i}$\ be the event that $\left\vert \widehat
{f}_{i}\left(  z_{i}\right)  \right\vert \geq1$\ and let $Y_{i}$\ be the event
that $\left\vert \widehat{f}_{i}\left(  z_{i}\right)  \right\vert \geq2$.
\ Also let $p_{i}:=\Pr\left[  X_{i}\right]  $\ and $q_{i}:=\Pr\left[
Y_{i}\right]  $, where the probability is over \texttt{FF-ALG}'s internal
(quantum) randomness. \ Then clearly%
\begin{align*}
p_{i}  &  =\frac{1}{N}\sum_{z_{i}:\left\vert \widehat{f}_{i}\left(
z_{i}\right)  \right\vert \geq1}\widehat{f}_{i}\left(  z_{i}\right)  ^{2},\\
q_{i}  &  =\frac{1}{N}\sum_{z_{i}:\left\vert \widehat{f}_{i}\left(
z_{i}\right)  \right\vert \geq2}\widehat{f}_{i}\left(  z_{i}\right)  ^{2}.
\end{align*}
So by assumption,%
\begin{align*}
p_{1}+\cdots+p_{n}  &  \geq0.8n,\\
q_{1}+\cdots+q_{n}  &  \geq0.26n.
\end{align*}
By a Chernoff/Hoeffding bound, it follows that%
\begin{align*}
\Pr\left[  X_{1}+\cdots+X_{n}\geq0.75n\right]   &  >1-\frac{1}{\exp\left(
n\right)  },\\
\Pr\left[  Y_{1}+\cdots+Y_{n}\geq0.25n\right]   &  >1-\frac{1}{\exp\left(
n\right)  }.
\end{align*}
Hence \texttt{FF-ALG} succeeds with $1-1/\exp\left(  n\right)  $\ probability
by the union bound.
\end{proof}

We also have the following:

\begin{lemma}
\label{mostgood}$\left\langle f_{1},\ldots,f_{n}\right\rangle $\ is good with
probability $1-1/\exp\left(  n\right)  $, if the $f_{i}$'s are chosen
uniformly at random.
\end{lemma}

\begin{proof}
Choose $f:\left\{  0,1\right\}  ^{n}\rightarrow\left\{  -1,1\right\}  $
uniformly at random. \ Then for each $z$, the Fourier coefficient $\widehat
{f}\left(  z\right)  $\ follows a normal distribution, with mean $0$ and
variance $1$. \ So in the limit of large $N$,%
\begin{align*}
\operatorname*{E}_{f}\left[  \sum_{z:\left\vert \widehat{f}\left(  z\right)
\right\vert \geq1}\widehat{f}\left(  z\right)  ^{2}\right]   &  =\sum
_{z\in\left\{  0,1\right\}  ^{n}}\Pr\left[  \left\vert \widehat{f}\left(
z\right)  \right\vert \geq1\right]  \operatorname*{E}\left[  \widehat
{f}\left(  z\right)  ^{2}~|~\left\vert \widehat{f}\left(  z\right)
\right\vert \geq1\right] \\
&  \approx\frac{2N}{\sqrt{2\pi}}\int_{1}^{\infty}e^{-x^{2}/2}x^{2}dx\\
&  \approx0.801N.
\end{align*}
Likewise,%
\begin{align*}
\operatorname*{E}_{f}\left[  \sum_{z:\left\vert \widehat{f}\left(  z\right)
\right\vert \geq2}\widehat{f}\left(  z\right)  ^{2}\right]   &  \approx
\frac{2N}{\sqrt{2\pi}}\int_{2}^{\infty}e^{-x^{2}/2}x^{2}dx\\
&  \approx0.261N.
\end{align*}
Since the $f_{i}$'s are chosen independently of one another, it follows by a
Chernoff bound that%
\begin{align*}
\sum_{i=1}^{n}\sum_{z_{i}:\left\vert \widehat{f}_{i}\left(  z_{i}\right)
\right\vert \geq1}\widehat{f}_{i}\left(  z_{i}\right)  ^{2}  &  \geq0.8Nn,\\
\sum_{i=1}^{n}\sum_{z_{i}:\left\vert \widehat{f}_{i}\left(  z_{i}\right)
\right\vert \geq2}\widehat{f}_{i}\left(  z_{i}\right)  ^{2}  &  \geq0.26Nn
\end{align*}
with probability $1-1/\exp\left(  n\right)  $ over the choice of $\left\langle
f_{1},\ldots,f_{n}\right\rangle $.
\end{proof}

Combining Lemmas \ref{ffinbqp} and \ref{mostgood}, we find that
\texttt{FF-ALG} succeeds with probability $1-1/\exp\left(  n\right)  $, where
the probability is over both $\left\langle f_{1},\ldots,f_{n}\right\rangle
$\ and \texttt{FF-ALG}'s internal randomness.

\subsection{\label{FCALG}Quantum Algorithm for \textsc{Fourier Checking}}

We now turn to \textsc{Fourier Checking}, the problem of deciding whether two
Boolean functions $f,g$\ are independent or forrelated. \ Here is a quantum
algorithm, \texttt{FC-ALG}, that solves \textsc{Fourier Checking}\ with
constant error probability using $O\left(  1\right)  $\ queries.\ \ First
prepare a uniform superposition over all $x\in\left\{  0,1\right\}  ^{n}$.
\ Then query $f$ in superposition, to create the state%
\[
\frac{1}{\sqrt{N}}\sum_{x\in\left\{  0,1\right\}  ^{n}}f\left(  x\right)
\left\vert x\right\rangle
\]
Then apply Hadamard gates to all $n$ qubits, to create the state%
\[
\frac{1}{N}\sum_{x,y\in\left\{  0,1\right\}  ^{n}}f\left(  x\right)  \left(
-1\right)  ^{x\cdot y}\left\vert y\right\rangle .
\]
Then query $g$ in superposition, to create the state%
\[
\frac{1}{N}\sum_{x,y\in\left\{  0,1\right\}  ^{n}}f\left(  x\right)  \left(
-1\right)  ^{x\cdot y}g\left(  y\right)  \left\vert y\right\rangle .
\]
Then apply Hadamard gates to all $n$ qubits again, to create the state%
\[
\frac{1}{N^{3/2}}\sum_{x,y,z\in\left\{  0,1\right\}  ^{n}}f\left(  x\right)
\left(  -1\right)  ^{x\cdot y}g\left(  y\right)  \left(  -1\right)  ^{y\cdot
z}\left\vert z\right\rangle .
\]
Finally, measure in the computational basis, and \textquotedblleft
accept\textquotedblright\ if and only if the outcome $\left\vert
0\right\rangle ^{\otimes n}$ is observed. \ If needed, repeat the whole
algorithm $O\left(  1\right)  $\ times to boost the success probability.

It is clear that the probability of observing $\left\vert 0\right\rangle
^{\otimes n}$ (in a single run of \texttt{FC-ALG})\ equals%
\[
p\left(  f,g\right)  :=\frac{1}{N^{3}}\left(  \sum_{x,y\in\left\{
0,1\right\}  ^{n}}f\left(  x\right)  \left(  -1\right)  ^{x\cdot y}g\left(
y\right)  \right)  ^{2}.
\]
Recall that \textsc{Promise Fourier Checking} was the problem of deciding
whether $p\left(  f,g\right)  \geq0.05$\ or $p\left(  f,g\right)  \leq0.01$,
promised that one of these is the case. \ Thus, we immediately get a quantum
algorithm to solve \textsc{Promise Fourier Checking}, with constant error
probability, using $O\left(  1\right)  $\ queries to $f$ and $g$.

For the distributional version of \textsc{Fourier Checking}, we also need the
following theorem.

\begin{theorem}
\label{fcinbqp}If $\left\langle f,g\right\rangle $\ is drawn from the uniform
distribution $\mathcal{U}$, then%
\[
\operatorname*{E}_{\mathcal{U}}\left[  p\left(  f,g\right)  \right]  =\frac
{1}{N}.
\]
If $\left\langle f,g\right\rangle $\ is drawn from the forrelated distribution
$\mathcal{F}$, then%
\[
\operatorname*{E}_{\mathcal{F}}\left[  p\left(  f,g\right)  \right]  >0.07.
\]

\end{theorem}

\begin{proof}
The first part follows immediately by symmetry (i.e., the fact that all
$N=2^{n}$ measurement outcomes of the quantum algorithm are equally likely).

For the second part, let $v\in\mathbb{R}^{N}$ be the vector of independent
Gaussians used to generate $f$ and $g$, let $w=v/\left\Vert v\right\Vert _{2}%
$\ be $v$ scaled to have unit norm, and let $H$ be the $n$-qubit Hadamard
matrix. \ Also let $\operatorname*{flat}\left(  w\right)  $\ be the unit
vector whose $x^{th}$\ entry is $\operatorname*{sgn}\left(  w_{x}\right)
/\sqrt{N}=f\left(  x\right)  /\sqrt{N}$, and let $\operatorname*{flat}\left(
Hw\right)  $\ be the unit vector whose $x^{th}$\ entry is $\operatorname*{sgn}%
\left(  \widehat{v}_{x}\right)  /\sqrt{N}=g\left(  x\right)  /\sqrt{N}$.
\ Then $p\left(  f,g\right)  $\ equals%
\[
\left(  \operatorname*{flat}\left(  w\right)  ^{T}H\operatorname*{flat}\left(
Hw\right)  \right)  ^{2},
\]
or the squared inner product between the vectors $\operatorname*{flat}\left(
w\right)  $ and $H\operatorname*{flat}\left(  Hw\right)  $. \ Note that
$w^{T}\cdot HHw=w^{T}w=1$. \ So the whole problem is to understand the
\textquotedblleft discretization error\textquotedblright\ incurred in
replacing $w^{T}$\ by $\operatorname*{flat}\left(  w\right)  ^{T}$\ and
$HHw$\ by $H\operatorname*{flat}\left(  Hw\right)  $. \ By the triangle
inequality, the angle between $\operatorname*{flat}\left(  w\right)  $\ and
$H\operatorname*{flat}\left(  Hw\right)  $\ is at most the angle between
$\operatorname*{flat}\left(  w\right)  $\ and $w$, plus the angle between
$w$\ and $H\operatorname*{flat}\left(  Hw\right)  $. \ In other words:%
\[
\arccos\left(  \operatorname*{flat}\left(  w\right)  ^{T}H\operatorname*{flat}%
\left(  Hw\right)  \right)  \leq\arccos\left(  \operatorname*{flat}\left(
w\right)  ^{T}w\right)  +\arccos\left(  w^{T}H\operatorname*{flat}\left(
Hw\right)  \right)  .
\]
Now,%
\begin{align*}
\operatorname*{flat}\left(  w\right)  ^{T}w  &  =\sum_{x\in\left\{
0,1\right\}  ^{n}}w_{x}\cdot\frac{1}{\sqrt{N}}\frac{\left\vert w_{x}%
\right\vert }{w_{x}}\\
&  =\frac{1}{\sqrt{N}}\sum_{x\in\left\{  0,1\right\}  ^{n}}\left\vert
w_{x}\right\vert \\
&  =\frac{\sum_{x\in\left\{  0,1\right\}  ^{n}}\left\vert v_{x}\right\vert
}{\sqrt{N}\left\Vert v\right\Vert _{2}}.
\end{align*}
Recall that each $v_{x}$\ is an independent real Gaussian with mean $0$ and
variance $1$,\ meaning that each $\left\vert v_{x}\right\vert $ is an
independent nonnegative random variable with expectation $\sqrt{2/\pi}$. \ So
by standard tail bounds, for all constants $\varepsilon>0$\ we have%
\begin{align*}
\Pr_{v}\left[  \sum_{x\in\left\{  0,1\right\}  ^{n}}\left\vert v_{x}%
\right\vert \leq\left(  \sqrt{\frac{2}{\pi}}-\varepsilon\right)  N\right]   &
\leq\frac{1}{\exp\left(  N\right)  },\\
\Pr\left[  \left\Vert v\right\Vert _{2}^{2}\geq\left(  1+\varepsilon\right)
N\right]   &  \leq\frac{1}{\exp\left(  N\right)  }.
\end{align*}
So by the union bound,%
\[
\Pr_{v}\left[  \operatorname*{flat}\left(  w\right)  ^{T}w\leq\sqrt{\frac
{2}{\pi}}-\varepsilon\right]  \leq\frac{1}{\exp\left(  N\right)  }.
\]
Since $H$ is unitary, the same analysis applies to $w^{T}H\operatorname*{flat}%
\left(  Hw\right)  $. \ Therefore, for all constants $\varepsilon>0$, with
$1-1/\exp\left(  N\right)  $ probability we have%
\begin{align*}
\arccos\left(  \operatorname*{flat}\left(  w\right)  ^{T}w\right)   &
\leq\left(  \arccos\sqrt{\frac{2}{\pi}}\right)  +\varepsilon,\\
\arccos\left(  w^{T}H\operatorname*{flat}\left(  Hw\right)  \right)   &
\leq\left(  \arccos\sqrt{\frac{2}{\pi}}\right)  +\varepsilon.
\end{align*}
So setting $\varepsilon=0.0001$,%
\begin{align*}
\arccos\left(  \operatorname*{flat}\left(  w\right)  ^{T}H\operatorname*{flat}%
\left(  Hw\right)  \right)   &  \leq\arccos\left(  \operatorname*{flat}\left(
w\right)  ^{T}w\right)  +\arccos\left(  w^{T}H\operatorname*{flat}\left(
Hw\right)  \right) \\
&  \leq2\left(  \arccos\sqrt{\frac{2}{\pi}}\right)  +2\varepsilon\\
&  \leq1.3
\end{align*}
Therefore, with $1-1/\exp\left(  N\right)  $ probability over $\left\langle
f,g\right\rangle $\ drawn from $\mathcal{F}$,%
\[
\operatorname*{flat}\left(  w\right)  ^{T}H\operatorname*{flat}\left(
Hw\right)  \geq\cos1.3,
\]
in which case $p\left(  f,g\right)  \geq\left(  \cos1.3\right)  ^{2}%
\approx0.0\allowbreak72$.
\end{proof}

Combining Theorem \ref{fcinbqp}\ with Markov's inequality, we immediately get
the following:

\begin{corollary}
\label{pfgcor}%
\begin{align*}
\Pr_{\left\langle f,g\right\rangle \sim\mathcal{U}}\left[  p\left(
f,g\right)  \geq0.01\right]   &  \leq\frac{100}{N},\\
\Pr_{\left\langle f,g\right\rangle \sim\mathcal{D}}\left[  p\left(
f,g\right)  \geq0.05\right]   &  \geq\frac{1}{50}.
\end{align*}

\end{corollary}

\section{\label{RELATIONAL}The Classical Complexity of \textsc{Fourier
Fishing}}

In Section \ref{FFALG}, we gave a quantum algorithm for \textsc{Fourier
Fishing}\ that made only one query to each $f_{i}$. \ By contrast, it is not
hard to show that any classical algorithm for \textsc{Fourier Fishing}%
\ requires exponentially many queries to the $f_{i}$'s. \ In this section, we
prove a much stronger result: that \textsc{Fourier Fishing}\ is not in
$\mathsf{FBPP}^{\mathsf{PH}}$. \ This result does not rely on any unproved conjectures.

\subsection{Constant-Depth Circuit Lower Bounds\label{CLB}}

Our starting point will be the following $\mathsf{AC}^{0}$\ lower bound, which
can be found in the book of H\aa stad \cite{hastad:book} for example.

\begin{theorem}
[\cite{hastad:book}]\label{majlb}Any depth-$d$ circuit that accepts all
$n$-bit\ strings of Hamming weight $n/2+1$, and rejects all strings of Hamming
weight $n/2$, has size $\exp\left(  \Omega\left(  n^{1/\left(  d-1\right)
}\right)  \right)  $.
\end{theorem}

We now give a corollary of Theorem \ref{majlb}, which (though simple) seems to
be new, and might be of independent interest. \ Consider the following
problem, which we call $\varepsilon$\textsc{-Bias Detection}. \ We are given a
string $y=y_{1}\ldots y_{m}\in\left\{  0,1\right\}  ^{m}$, and are promised
that each bit $y_{i}$\ is $1$ with independent probability $p$. \ The task is
to decide whether $p=1/2$\ or $p=1/2+\varepsilon$.

\begin{corollary}
\label{ac0bias}Let $\mathcal{U}\left[  \varepsilon\right]  $\ be the
distribution over $\left\{  0,1\right\}  ^{m}$ where each bit is $1$\ with
independent probability $1/2+\varepsilon$. \ Then any depth-$d$ circuit $C$
such that%
\[
\left\vert \Pr_{\mathcal{U}\left[  \varepsilon\right]  }\left[  C\right]
-\Pr_{\mathcal{U}\left[  0\right]  }\left[  C\right]  \right\vert
=\Omega\left(  1\right)
\]
has size $\exp\left(  \Omega\left(  1/\varepsilon^{1/\left(  d+2\right)
}\right)  \right)  $.
\end{corollary}

\begin{proof}
Suppose such a distinguishing circuit $C$ exists, with depth $d$ and size $S$,
for some $\varepsilon>0$ (the parameter $m$ is actually irrelevant). \ Let
$n=1/\varepsilon$, and assume for simplicity that $n$ is an integer. \ Using
$C$, we will construct a new circuit $C^{\prime}$\ with depth $d^{\prime}%
=d+3$\ and size $S^{\prime}=O\left(  nS\right)  +\operatorname*{poly}\left(
n\right)  $, which accepts all strings $x\in\left\{  0,1\right\}  ^{n}$\ of
Hamming weight $n/2+1$, and rejects all strings of Hamming weight $n/2$. \ By
Theorem \ref{majlb}, this will imply that the original circuit $C$ must have
had size%
\begin{align*}
S  &  =\frac{1}{n}\exp\left(  \Omega\left(  n^{1/\left(  d^{\prime}-1\right)
}\right)  \right)  -\operatorname*{poly}\left(  n\right) \\
&  =\exp\left(  \Omega\left(  1/\varepsilon^{1/\left(  d+2\right)  }\right)
\right)  .
\end{align*}
So fix an input $x\in\left\{  0,1\right\}  ^{n}$, and suppose we choose $m$
bits $x_{i_{1}},\ldots,x_{i_{m}}$ from $x$, with each index $i_{j}$\ chosen
uniformly at random with replacement. \ Call the resulting $m$-bit string $y$.
\ Observe that if $x$\ had Hamming weight $n/2$, then $y$ will be distributed
according to $\mathcal{U}\left[  0\right]  $, while if $x$ had Hamming weight
$n/2+1$, then $y$ will be distributed according to $\mathcal{U}\left[
\varepsilon\right]  $. \ So by assumption,%
\begin{align*}
\Pr\left[  C\left(  y\right)  ~|~\left\vert x\right\vert =n/2\right]   &
=\alpha,\\
\Pr\left[  C\left(  y\right)  ~|~\left\vert x\right\vert =n/2+1\right]   &
=\alpha+\delta
\end{align*}
for some constants $\alpha$ and $\delta\neq0\ $(we can assume $\delta
>0$\ without loss of generality).

Now suppose we repeat the above experiment $T=kn$\ times, for some constant
$k=k\left(  \alpha,\delta\right)  $. \ That is, we create $T$ strings
$y_{1},\ldots,y_{T}$ by choosing random bits of $x$, so that each $y_{i}$\ is
distributed independently according to either $\mathcal{U}\left[  0\right]
$\ or $\mathcal{U}\left[  \varepsilon\right]  $. \ We then apply $C$ to each
$y_{i}$. \ Let%
\[
Z=C\left(  y_{1}\right)  +\cdots+C\left(  y_{T}\right)
\]
be the number of $C$ invocations that accept. \ Then by a Chernoff bound, if
$\left\vert x\right\vert =n/2$\ then%
\[
\Pr\left[  Z>\alpha T+\frac{\delta}{3}T\right]  <\exp\left(  -n\right)  ,
\]
while if $\left\vert x\right\vert =n/2+1$\ then%
\[
\Pr\left[  Z<\alpha T+\frac{2\delta}{3}T\right]  <\exp\left(  -n\right)  .
\]
By taking $k$ large enough, we can make both of these probabilities less than
$2^{-n}$. \ By the union bound, this implies that there must exist a way to
choose $y_{1},\ldots,y_{T}$\ so that%
\begin{align*}
\left\vert x\right\vert =\frac{n}{2}  &  \Longrightarrow Z\leq\alpha
T+\frac{\delta}{3}T,\\
\left\vert x\right\vert =\frac{n}{2}+1  &  \Longrightarrow Z\geq\alpha
T+\frac{2\delta}{3}T
\end{align*}
for \textit{every} $x$\ with $\left\vert x\right\vert \in\left\{
n/2,n/2+1\right\}  $\ simultaneously. \ In forming the circuit $C^{\prime}$,
we simply hardwire that choice.

The last step is to decide whether $Z\leq\alpha T+\frac{\delta}{3}T$\ or
$Z\geq\alpha T+\frac{2\delta}{3}T$. \ This can be done using an $\mathsf{AC}%
^{0}$\ circuit for the \textsc{Approximate Majority} problem (see Viola
\cite{viola:maj} for example), which has depth $3$ and size
$\operatorname*{poly}\left(  T\right)  $. \ The end result is a circuit
$C^{\prime}$\ to distinguish $\left\vert x\right\vert =n/2$\ from $\left\vert
x\right\vert =n/2+1$, which has depth $d+3$\ and size $TS+\operatorname*{poly}%
\left(  T\right)  =O\left(  nS\right)  +\operatorname*{poly}\left(  n\right)
$.
\end{proof}

\subsection{Secretly Biased Fourier Coefficients\label{SBFC}}

In this section, we prove two lemmas indicating that one can \textit{slightly}
bias one of the Fourier coefficients of a random Boolean function $f:\left\{
0,1\right\}  ^{n}\rightarrow\left\{  -1,1\right\}  $, and yet still have $f$
be information-theoretically indistinguishable from a random Boolean function
(so that, in particular, an adversary has no way of knowing which Fourier
coefficient was biased). \ These lemmas will play a key role in our reduction
from $\varepsilon$\textsc{-Bias Detection}\ to \textsc{Fourier Fishing}.

Fix a string $s\in\left\{  0,1\right\}  ^{n}$. \ Let $\mathcal{A}\left[
s\right]  $\ be the probability distribution over functions $f:\left\{
0,1\right\}  ^{n}\rightarrow\left\{  -1,1\right\}  $ where each $f\left(
x\right)  $\ is $1$\ with independent probability $\frac{1}{2}+\left(
-1\right)  ^{s\cdot x}\frac{1}{2\sqrt{N}}$, and let $\mathcal{B}\left[
s\right]  $\ be the distribution where each $f\left(  x\right)  $\ is
$1$\ with independent probability $\frac{1}{2}-\left(  -1\right)  ^{s\cdot
x}\frac{1}{2\sqrt{N}}$. \ Then let $\mathcal{D}\left[  s\right]  =\frac{1}%
{2}\left(  \mathcal{A}\left[  s\right]  +\mathcal{B}\left[  s\right]  \right)
$\ (that is, an equal mixture of $\mathcal{A}\left[  s\right]  $\ and
$\mathcal{B}\left[  s\right]  $).

\begin{lemma}
\label{alicebob}Suppose Alice chooses $s\in\left\{  0,1\right\}  ^{n}%
$\ uniformly at random, then draws $f$\ according to $\mathcal{D}\left[
s\right]  $. \ She keeps $s$ secret, but sends the truth table of $f$ to Bob.
\ After examining $f$, Bob outputs a string $z$\ such that $\left\vert
\widehat{f}\left(  z\right)  \right\vert \geq\beta$. \ Then%
\[
\Pr\left[  s=z\right]  \geq\frac{e^{\beta}+e^{-\beta}}{2\sqrt{e}N}.
\]
where the probability is over all runs of the protocol.
\end{lemma}

\begin{proof}
By Yao's principle, we can assume without loss of generality that Bob's
strategy is deterministic. \ For each $z$, let $\mathcal{F}\left[  z\right]
$\ be the set of all $f$'s that cause Bob to output $z$. \ Then the first step
is to lower-bound $\Pr_{\mathcal{D}\left[  z\right]  }\left[  f\right]  $, for
some fixed $z$\ and $f\in\mathcal{F}\left[  z\right]  $. \ Let $N_{f}\left[
z\right]  $\ be the number of inputs $x\in\left\{  0,1\right\}  ^{n}$\ such
that $f\left(  x\right)  =\left(  -1\right)  ^{z\cdot x}$. \ It is not hard to
see that $N_{f}\left[  z\right]  =\frac{N}{2}+\frac{\sqrt{N}\widehat{f}\left(
z\right)  }{2}$. \ So%
\begin{align*}
\Pr_{\mathcal{D}\left[  z\right]  }\left[  f\right]   &  =\frac{1}{2}\left(
\Pr_{\mathcal{A}\left[  z\right]  }\left[  f\right]  +\Pr_{\mathcal{B}\left[
z\right]  }\left[  f\right]  \right) \\
&  =\frac{1}{2}\left(
%TCIMACRO{\dprod \limits_{x\in\left\{  0,1\right\}  ^{n}}}%
%BeginExpansion
{\displaystyle\prod\limits_{x\in\left\{  0,1\right\}  ^{n}}}
%EndExpansion
\left(  \frac{1}{2}+\frac{\left(  -1\right)  ^{z\cdot x}f\left(  x\right)
}{2\sqrt{N}}\right)  +%
%TCIMACRO{\dprod \limits_{x\in\left\{  0,1\right\}  ^{n}}}%
%BeginExpansion
{\displaystyle\prod\limits_{x\in\left\{  0,1\right\}  ^{n}}}
%EndExpansion
\left(  \frac{1}{2}-\frac{\left(  -1\right)  ^{z\cdot x}f\left(  x\right)
}{2\sqrt{N}}\right)  \right) \\
&  =\frac{1}{2^{N+1}}\left(  \left(  1+\frac{1}{\sqrt{N}}\right)
^{N_{f}\left[  z\right]  }\left(  1-\frac{1}{\sqrt{N}}\right)  ^{N-N_{f}%
\left[  z\right]  }+\left(  1-\frac{1}{\sqrt{N}}\right)  ^{N_{f}\left[
z\right]  }\left(  1+\frac{1}{\sqrt{N}}\right)  ^{N-N_{f}\left[  z\right]
}\right) \\
&  =\frac{1}{2^{N+1}}\left(  \frac{\left(  1+1/\sqrt{N}\right)  ^{\left(
\sqrt{N}\widehat{f}\left(  z\right)  +N\right)  /2}}{\left(  1-1/\sqrt
{N}\right)  ^{\left(  \sqrt{N}\widehat{f}\left(  z\right)  -N\right)  /2}%
}+\frac{\left(  1-1/\sqrt{N}\right)  ^{\left(  \sqrt{N}\widehat{f}\left(
z\right)  +N\right)  /2}}{\left(  1+1/\sqrt{N}\right)  ^{\left(  \sqrt
{N}\widehat{f}\left(  z\right)  -N\right)  /2}}\right) \\
&  =\frac{1}{2^{N+1}}\left(  1-\frac{1}{N}\right)  ^{N/2}\left(  \left(
\frac{1+1/\sqrt{N}}{1-1/\sqrt{N}}\right)  ^{\sqrt{N}\widehat{f}\left(
z\right)  /2}+\left(  \frac{1-1/\sqrt{N}}{1+1/\sqrt{N}}\right)  ^{\sqrt
{N}\widehat{f}\left(  z\right)  /2}\right) \\
&  =\frac{1}{2\sqrt{e}2^{N}}\left(  e^{\widehat{f}\left(  z\right)
}+e^{-\widehat{f}\left(  z\right)  }\right) \\
&  \geq\frac{e^{\beta}+e^{-\beta}}{2\sqrt{e}2^{N}}.
\end{align*}
Here the second-to-last line takes the limit as $N\rightarrow\infty$, while
the last line follows from the assumption $\left\vert \widehat{f}\left(
z\right)  \right\vert \geq\beta$, together with the fact that $e^{y}+e^{-y}%
$\ increases monotonically away from $y=0$.

Summing over all $z$ and $f$,%
\begin{align*}
\Pr\left[  s=z\right]   &  =\sum_{z\in\left\{  0,1\right\}  ^{n}}\sum
_{f\in\mathcal{F}\left[  z\right]  }\Pr\left[  f\right]  \cdot\Pr\left[
s=z~\ |~\ f\right] \\
&  =\sum_{z\in\left\{  0,1\right\}  ^{n}}\sum_{f\in\mathcal{F}\left[
z\right]  }\Pr\left[  f\right]  \cdot\frac{\Pr\left[  f~\ |~\ s=z\right]
\Pr\left[  s=z\right]  }{\Pr\left[  f\right]  }\\
&  =\frac{1}{N}\sum_{z\in\left\{  0,1\right\}  ^{n}}\sum_{f\in\mathcal{F}%
\left[  z\right]  }\Pr_{\mathcal{D}\left[  z\right]  }\left[  f\right] \\
&  \geq\frac{e^{\beta}+e^{-\beta}}{2\sqrt{e}N}.
\end{align*}

\end{proof}

Now let $\mathcal{D}=\operatorname*{E}_{s}\left[  \mathcal{D}\left[  s\right]
\right]  $ (that is, an equal mixture of all the $\mathcal{D}\left[  s\right]
$'s). \ We claim that $\mathcal{D}$\ is extremely close in variation distance
to $\mathcal{U}$,\ the uniform distribution over all Boolean functions
$f:\left\{  0,1\right\}  ^{n}\rightarrow\left\{  -1,1\right\}  $.

\begin{lemma}
\label{closetounif}$\left\Vert \mathcal{D}-\mathcal{U}\right\Vert \leq
\frac{e-1}{2\sqrt{2eN}}$.
\end{lemma}

\begin{proof}
By a calculation from Lemma \ref{alicebob}, for all $f$ and $s$ we have%
\[
\Pr_{\mathcal{D}\left[  s\right]  }\left[  f\right]  =\frac{1}{2\sqrt{e}2^{N}%
}\left(  e^{\widehat{f}\left(  s\right)  }+e^{-\widehat{f}\left(  s\right)
}\right)
\]
in the limit of large $N$. \ Hence%
\[
\Pr_{\mathcal{D}}\left[  f\right]  =\operatorname*{E}_{s}\left[
\Pr_{\mathcal{D}\left[  s\right]  }\left[  f\right]  \right]  =\frac{1}%
{2\sqrt{e}N2^{N}}\sum_{s\in\left\{  0,1\right\}  ^{n}}\left(  e^{\widehat
{f}\left(  s\right)  }+e^{-\widehat{f}\left(  s\right)  }\right)  .
\]
Clearly $\operatorname*{E}_{f}\left[  \Pr_{\mathcal{D}}\left[  f\right]
\right]  =1/2^{N}$. \ Our goal is to upper-bound the variance
$\operatorname*{Var}_{f}\left[  \Pr_{\mathcal{D}}\left[  f\right]  \right]  $,
which measures the distance from $\mathcal{D}$\ to the uniform distribution.
\ In the limit of large $N$, we have%
\begin{align*}
\operatorname*{E}_{f}\left[  \Pr_{\mathcal{D}}\left[  f\right]  ^{2}\right]
&  =\frac{1}{4eN^{2}2^{2N}}\left(  \sum_{s}\operatorname*{E}_{f}\left[
\left(  e^{\widehat{f}\left(  s\right)  }+e^{-\widehat{f}\left(  s\right)
}\right)  ^{2}\right]  +\sum_{s\neq t}\operatorname*{E}_{f}\left[  \left(
e^{\widehat{f}\left(  s\right)  }+e^{-\widehat{f}\left(  s\right)  }\right)
\left(  e^{\widehat{f}\left(  t\right)  }+e^{-\widehat{f}\left(  t\right)
}\right)  \right]  \right) \\
&  =\frac{1}{4eN^{2}2^{2N}}\left(
\begin{array}
[c]{c}%
\sum_{s}\frac{1}{\sqrt{2\pi}}\int_{-\infty}^{\infty}e^{-x^{2}/2}\left(
e^{x}+e^{-x}\right)  ^{2}dx\\
+\sum_{s\neq t}\left[  \frac{1}{\sqrt{2\pi}}\int_{-\infty}^{\infty}%
e^{-x^{2}/2}\left(  e^{x}+e^{-x}\right)  dx\right]  ^{2}%
\end{array}
\right) \\
&  =\frac{1}{4eN^{2}2^{2N}}\left[  \left(  2e^{2}+2\right)  N+4eN\left(
N-1\right)  \right] \\
&  =\frac{1}{2^{2N}}\left(  1+\frac{\left(  e-1\right)  ^{2}}{2eN}\right)  .
\end{align*}
Hence%
\[
\operatorname*{Var}_{f}\left[  \Pr_{\mathcal{D}}\left[  f\right]  \right]
=\operatorname*{E}_{f}\left[  \Pr_{\mathcal{D}}\left[  f\right]  ^{2}\right]
-\operatorname*{E}_{f}\left[  \Pr_{\mathcal{D}}\left[  f\right]  \right]
^{2}=\frac{\left(  e-1\right)  ^{2}}{2eN2^{2N}}.
\]
So by Cauchy-Schwarz,%
\[
\operatorname*{E}_{f}\left[  \left\vert \Pr_{\mathcal{D}}\left[  f\right]
-\Pr_{\mathcal{U}}\left[  f\right]  \right\vert \right]  \leq\sqrt
{\operatorname*{Var}_{f}\left[  \Pr_{\mathcal{D}}\left[  f\right]  \right]
}=\frac{e-1}{\sqrt{2eN}}\cdot\frac{1}{2^{N}}%
\]
and%
\[
\left\Vert \mathcal{D}-\mathcal{U}\right\Vert \leq\frac{e-1}{2\sqrt{2eN}}.
\]

\end{proof}

An immediate corollary of Lemma \ref{closetounif}\ is that, if a
\textsc{Fourier Fishing} algorithm succeeds with probability $p$ on
$\left\langle f_{1},\ldots,f_{n}\right\rangle $ drawn from $\mathcal{U}^{n}$,
then it also succeeds with probability at least%
\[
p-\left\Vert \mathcal{D}^{n}-\mathcal{U}^{n}\right\Vert \geq p-\frac{\left(
e-1\right)  n}{2\sqrt{2eN}}%
\]
on $\left\langle f_{1},\ldots,f_{n}\right\rangle $ drawn from $\mathcal{D}%
^{n}$.

\subsection{Putting It All Together\label{ALLTOGETHER}}

Using the results of Sections \ref{CLB}\ and \ref{SBFC}, we are now ready to
prove a lower bound on the constant-depth circuit complexity of
\textsc{Fourier Fishing}.

\begin{theorem}
\label{fflb}Any depth-$d$ circuit that solves the \textsc{Fourier Fishing}
problem, with probability at least $0.99$\ over $f_{1},\ldots,f_{n}$ chosen
uniformly at random,\ has size $\exp\left(  \Omega\left(  N^{1/\left(
2d+8\right)  }\right)  \right)  $.
\end{theorem}

\begin{proof}
Let $C$ be a circuit of depth $d$ and size $s$. \ Let $G$\ be the set of all
$\left\langle f_{1},\ldots,f_{n}\right\rangle $ on which $C$ \textit{succeeds}%
: that is, for which it outputs $z_{1},\ldots,z_{n}$, at least 75\% of which
satisfy $\left\vert \widehat{f}_{i}\left(  z_{i}\right)  \right\vert \geq1$
and at least 25\% of which satisfy $\left\vert \widehat{f}_{i}\left(
z_{i}\right)  \right\vert \geq2$. \ Suppose%
\[
\Pr_{\mathcal{U}^{n}}\left[  \left\langle f_{1},\ldots,f_{n}\right\rangle \in
G\right]  \geq0.99.
\]
Then by Lemma \ref{closetounif}, we also have%
\[
\Pr_{\mathcal{D}^{n}}\left[  \left\langle f_{1},\ldots,f_{n}\right\rangle \in
G\right]  \geq0.99-\frac{\left(  e-1\right)  n}{2\sqrt{2eN}}\geq0.98
\]
for sufficiently large $n$.

Using the above fact, we will convert $C$ into a new circuit $C^{\prime}%
$\ that solves the $\varepsilon$\textsc{-Bias Detection} problem of Corollary
\ref{ac0bias}, with $\varepsilon:=\frac{1}{2\sqrt{N}}$. \ This $C^{\prime}$
will have depth $d^{\prime}=d+2$\ and size $S^{\prime}=O\left(  NS\right)  $.
\ By Corollary \ref{ac0bias}, this will imply that $C$ itself must have had
size%
\begin{align*}
S  &  =\exp\left(  \Omega\left(  1/\varepsilon^{1/\left(  d^{\prime}+2\right)
}\right)  \right) \\
&  =\exp\left(  \Omega\left(  N^{1/\left(  2d+8\right)  }\right)  \right)  .
\end{align*}
Let $M=N^{2}n$, and let $R=r_{1}\ldots r_{M}\in\left\{  0,1\right\}  ^{M}$ be
a string of bits where each $r_{j}$\ is $1$ with independent probability $p$.
\ We want to decide whether $p=1/2$\ or $p=1/2+\varepsilon$---that is, whether
$R$ was drawn from $\mathcal{U}\left[  0\right]  $\ or $\mathcal{U}\left[
\varepsilon\right]  $. \ We can do this as follows.\ \ First, choose strings
$s_{1},\ldots,s_{n}\in\left\{  0,1\right\}  ^{n}$, bits $b_{1},\ldots,b_{n}%
\in\left\{  0,1\right\}  $, and an integer $k\in\left[  n\right]  $\ uniformly
at random. \ Next, define Boolean functions $f_{1},\ldots,f_{n}:\left\{
0,1\right\}  ^{n}\rightarrow\left\{  -1,1\right\}  $\ using the first $Nn$
bits of $R$, like so:%
\[
f_{i}\left(  x\right)  :=\left(  -1\right)  ^{r_{\left(  i-1\right)
N+x}+s_{i}\cdot x+b_{i}}.
\]
Finally, feed $\left\langle f_{1},\ldots,f_{n}\right\rangle $ as input to $C$,
and consider $z_{k}$, the $k^{th}$ output of $C$ (discarding the other
$n-1$\ outputs). \ We are interested in $\Pr\left[  z_{k}=s_{k}\right]
$,\ where the probability is over $R$, $s_{1},\ldots,s_{n}$, $b_{1}%
,\ldots,b_{n}$, and $k$.

If $p=1/2$, notice that $f_{1},\ldots,f_{n}$ are independent and uniformly
random regardless of $s_{1},\ldots,s_{n}$. \ So $C$ gets no information about
$s_{k}$, and $\Pr\left[  z_{k}=s_{k}\right]  =1/N$.

On the other hand, if $p=1/2+\varepsilon$, then each $f_{i}$\ is drawn
independently from the distribution $\mathcal{D}\left[  s\right]  $\ studied
in Lemma \ref{alicebob}. \ So by the Lemma, for every $i\in\left[  n\right]
$, if $\left\vert \widehat{f}_{i}\left(  z_{i}\right)  \right\vert \geq\beta
$\ then%
\[
\Pr_{f_{i}}\left[  z_{i}=s_{i}\right]  \geq\frac{e^{\beta}+e^{-\beta}}%
{2\sqrt{e}N}.
\]
So assuming $C$\ succeeds (that is, $\left\langle f_{1},\ldots,f_{n}%
\right\rangle \in G$), we have%
\[
\Pr_{f_{1},\ldots,f_{n},k}\left[  z_{k}=s_{k}\right]  \geq\frac{1}{4}\left(
\frac{e^{2}+e^{-2}}{2\sqrt{e}N}\right)  +\frac{1}{2}\left(  \frac{e^{1}%
+e^{-1}}{2\sqrt{e}N}\right)  \geq\frac{1.038}{N}.
\]
So for a \textit{random} $\left\langle f_{1},\ldots,f_{n}\right\rangle $ drawn
according to $\mathcal{D}^{n}$,%
\[
\Pr_{f_{1},\ldots,f_{n},k}\left[  z_{k}=s_{k}\right]  \geq0.98\left(
\frac{1.038}{N}\right)  \geq\frac{1.017}{N}.
\]
Notice that this is bounded above $1/N$\ by a multiplicative constant.

Now let us repeat the above experiment $N$ times. \ That is, for all
$j:=1$\ to $N$, we generate\ Boolean functions $f_{j1},\ldots,f_{jn}:\left\{
0,1\right\}  ^{n}\rightarrow\left\{  -1,1\right\}  $ by the same probabilistic
procedure as before, but each time using a new $Nn$-bit substring of $R_{j}$
of $R$, as well as new $s$, $b$, and $k$ values (denoted $s_{j1},\ldots
,s_{jn}$, $b_{j1},\ldots,b_{jn}$, and $k_{j}$). \ We then apply $C$\ to each
$n$-tuple $\left\langle f_{j1},\ldots,f_{jn}\right\rangle $. \ Let $z_{j}$\ be
the $k_{j}^{th}$\ string that $C$ outputs when run on $\left\langle
f_{j1},\ldots,f_{jn}\right\rangle $. \ Then by the above, for each
$j\in\left[  N\right]  $ we have%
\begin{align*}
p=\frac{1}{2}  &  \Longrightarrow\Pr\left[  z_{j}=s_{jk_{j}}\right]  =\frac
{1}{N},\\
p=\frac{1}{2}+\varepsilon &  \Longrightarrow\Pr\left[  z_{j}=s_{jk_{j}%
}\right]  \geq\frac{1.017}{N}.
\end{align*}
Furthermore, these probabilities are independent across the different $j$'s.
\ So let $E$ be the event that there \textit{exists} a $j\in\left[  N\right]
$\ such that $z_{j}=s_{jk_{j}}$. \ Then if $p=1/2$ we have%
\[
\Pr\left[  E\right]  =1-\left(  1-\frac{1}{N}\right)  ^{N}\approx1-\frac{1}%
{e}\leq0.633,
\]
while if $p=1/2+\varepsilon$\ we have%
\[
\Pr\left[  E\right]  \geq1-\left(  1-\frac{1.017}{N}\right)  ^{N}\geq0.638.
\]

It should now be clear how to create the circuit $C^{\prime}$, which
distinguishes\ $R\in\left\{  0,1\right\}  ^{M}$ drawn from $\mathcal{U}\left[
0\right]  $\ from $R$ drawn from $\mathcal{U}\left[  \varepsilon\right]
$\ with constant bias. \ For each $j\in\left[  N\right]  $, generate an
$n$-tuple of Boolean functions $\left\langle f_{j1},\ldots,f_{jn}\right\rangle
$ from $R$ and apply $C$ to it; then check whether there exists a $j\in\left[
N\right]  $ such that $z_{j}=s_{jk_{j}}$. \ This checking step can be done by
a depth-$2$ circuit of size $O\left(  Nn\right)  $. \ Therefore, $C^{\prime}$
will have depth $d^{\prime}=d+2$\ and size $s^{\prime}=O\left(  Ns\right)  $.
\ A technicality is that our choices of the\ $s_{ji}$'s, $b_{ji}$'s, and
$k_{j}$'s were made randomly. \ However, by Yao's principle, there clearly
\textit{exist} $s_{ji}$'s, $b_{ji}$'s, and $k_{j}$'s such that%
\[
\Pr_{\mathcal{U}\left[  \varepsilon\right]  }\left[  C^{\prime}\left(
R\right)  \right]  -\Pr_{\mathcal{U}\left[  0\right]  }\left[  C^{\prime
}\left(  R\right)  \right]  \geq0.638-0.633=0.005.
\]
So in forming $C^{\prime}$, we simply hardwire those choices.
\end{proof}

Combining Theorem \ref{fflb}\ with standard diagonalization tricks, we can now
prove an oracle separation (in fact, a \textit{random} oracle separation)
between the complexity classes $\mathsf{FBQP}$\ and $\mathsf{FBPP}%
^{\mathsf{PH}}$.

\begin{theorem}
\label{fbqpsep}$\mathsf{FBQP}^{A}\not \subset \mathsf{FBPP}^{\mathsf{PH}^{A}}$
with probability $1$ for a random oracle $A$.
\end{theorem}

\begin{proof}
We interpret the oracle $A$ as encoding $n$ random Boolean functions
$f_{n1},\ldots,f_{nn}:\left\{  0,1\right\}  ^{n}\rightarrow\left\{
-1,1\right\}  $ for each positive integer $n$. \ Let $R$\ be the relational
problem where we are given $0^{n}$\ as input, and \textit{succeed} if and only
if we output strings $z_{1},\ldots,z_{n}\in\left\{  0,1\right\}  ^{n}$, at
least $3/4$ of which satisfy $\left\vert \widehat{f}_{ni}\left(  z_{i}\right)
\right\vert \geq1$ and at least $1/4$\ of which satisfy $\left\vert
\widehat{f}_{ni}\left(  z_{i}\right)  \right\vert \geq2$. \ Then by Lemmas
\ref{ffinbqp} and \ref{mostgood}, there exists an $\mathsf{FBQP}^{A}$ machine
$M$ such that for all $n$,%
\[
\Pr\left[  M\left(  0^{n}\right)  ~\text{succeeds}\right]  \geq1-\frac{1}%
{\exp\left(  n\right)  },
\]
where the probability is over both $A$\ and the quantum randomness. \ Hence
$\Pr\left[  M\left(  0^{n}\right)  ~\text{succeeds}\right]  \geq
1-1/\exp\left(  n\right)  $\ on all but finitely many $n$, with probability
$1$ over $A$. \ Since we can simply hardwire the answers on the $n$'s for
which $M$ fails, it follows that $R\in\mathsf{FBQP}^{A}$ with probability $1$
over $A$.

On the other hand, let $M$\ be an\ $\mathsf{FBPP}^{\mathsf{PH}^{A}}$\ machine.
\ Then by the standard conversion between $\mathsf{PH}$\ and $\mathsf{AC}^{0}%
$, for every $n$ there exists a probabilistic $\mathsf{AC}^{0}$\ circuit
$C_{M,n}$, of size $2^{\operatorname*{poly}\left(  n\right)  }%
=2^{\operatorname*{polylog}\left(  N\right)  }$, that takes $A$ as input and
simulates $M\left(  0^{n}\right)  $. \ By Yao's principle, we can assume
without loss of generality that $C_{M,n}$\ is deterministic, since the oracle
$A$ is already random. \ Then by Theorem \ref{fflb},%
\[
\Pr_{A}\left[  C_{M,n}\text{ succeeds}\right]  <0.99
\]
for all sufficiently large $n$. \ By the independence of the\ $f_{ni}$'s, this
is true even if we condition on $C_{M,1},\ldots,C_{M,n-1}$ succeeding. \ So as
in the standard random oracle argument of Bennett and Gill \cite{bg}, for
every fixed $M$ we have%
\[
\Pr_{A}\left[  C_{M,1},C_{M,2},C_{M,3},\ldots\text{ succeed}\right]  =0.
\]
So by the union bound,%
\[
\Pr_{A}\left[  \exists M:C_{M,1},C_{M,2},C_{M,3},\ldots\text{ succeed}\right]
=0
\]
as well. \ It follows that $\mathsf{FBQP}^{A}\not \subset \mathsf{FBPP}%
^{\mathsf{PH}^{A}}$\ with probability $1$ over $A$.
\end{proof}

If we \textquotedblleft scale down by an exponential,\textquotedblright\ then
we can eliminate the need for the oracle $A$, and get a relation problem that
is solvable in quantum \textit{logarithmic} time but not in $\mathsf{AC}^{0}$.

\begin{theorem}
\label{logtime}There exists a relation problem solvable in $\mathsf{BQLOGTIME}%
$\ but not in $\mathsf{AC}^{0}$.
\end{theorem}

\begin{proof}
In our relation problem $R$, the input (of size $M=2^{n}n$) will encode the
truth tables of $n$ Boolean functions, $f_{1},\ldots,f_{n}:\left\{
0,1\right\}  ^{n}\rightarrow\left\{  -1,1\right\}  $, which are promised to be
\textquotedblleft good\textquotedblright\ as defined in Section \ref{PROBLEMS}%
. \ The task is to solve \textsc{Promise Fourier Fishing} on $\left\langle
f_{1},\ldots,f_{n}\right\rangle $.

By Lemma \ref{ffinbqp}, there exists a quantum algorithm that runs in
$O\left(  n\right)  =O\left(  \log M\right)  $\ time, making random accesses
to the truth tables of $f_{1},\ldots,f_{n}$, that solves $R$ with probability
$1-1/\exp\left(  n\right)  =1-1/M^{\Omega\left(  1\right)  }$.

On the other hand, suppose $R$ is in $\mathsf{AC}^{0}$. \ Then we get a
nonuniform circuit family $\left\{  C_{n}\right\}  _{n}$, of depth $O\left(
1\right)  $\ and size $\operatorname*{poly}\left(  M\right)  =2^{O\left(
n\right)  }$,\ that solves \textsc{Fourier Fishing}\ on all tuples
$\left\langle f_{1},\ldots,f_{n}\right\rangle $\ that are good. \ Recall that
by Lemma \ref{mostgood}, a $1-1/\exp\left(  n\right)  $\ fraction of
$\left\langle f_{1},\ldots,f_{n}\right\rangle $'s are good. \ Therefore
$\left\{  C_{n}\right\}  _{n}$\ actually solves \textsc{Fourier Fishing}\ with
probability\ $1-1/\exp\left(  n\right)  $\ on $\left\langle f_{1},\ldots
,f_{n}\right\rangle $\ chosen uniformly at random. \ But this contradicts
Theorem \ref{fflb}.

Hence $R\in\mathsf{FBQLOGTIME}\setminus\mathsf{FAC}^{0}$ (where
$\mathsf{FBQLOGTIME}$\ and $\mathsf{FAC}^{0}$\ are the relation versions of
$\mathsf{BQLOGTIME}$\ and $\mathsf{AC}^{0}$\ respectively).
\end{proof}

\section{The Classical Complexity of \textsc{Fourier Checking}\label{CCFC}}

Section \ref{RELATIONAL}\ settled the relativized $\mathsf{BQP}$\ versus
$\mathsf{PH}$ question, if we are willing to talk about relation problems.
\ Ultimately, though, we also care about decision problems. \ So in this
section we consider the \textsc{Fourier Checking}\ problem, of deciding
whether two Boolean functions $f,g$\ are independent or forrelated. \ In
Section \ref{FCALG}, we saw that \textsc{Fourier Checking}\ has quantum query
complexity $O\left(  1\right)  $. \ What is its classical query
complexity?\footnote{So long as we consider the distributional version of
\textsc{Fourier Checking}, the deterministic and randomized query complexities
are the same (by Yao's principle).}

It is not hard to give a classical algorithm that solves \textsc{Fourier
Checking} using $O\left(  \sqrt{N}\right)  =O\left(  2^{n/2}\right)
$\ queries. \ The algorithm is as follows:\ for some $K=\Theta\left(  \sqrt
{N}\right)  $, first choose sets $X=\left\{  x_{1},\ldots,x_{K}\right\}  $ and
$Y=\left\{  y_{1},\ldots,y_{K}\right\}  $\ of $n$-bit strings uniformly at
random. \ Then query $f\left(  x_{i}\right)  $\ and $g\left(  y_{i}\right)  $
for all $i\in\left[  K\right]  $. \ Finally, compute%
\[
Z:=\sum_{i,j=1}^{K}f\left(  x_{i}\right)  \left(  -1\right)  ^{x_{i}\cdot
y_{j}}g\left(  y_{j}\right)  ,
\]
accept if $\left\vert Z\right\vert $ is greater than some cutoff $cK$, and
reject otherwise. \ For suitable $K$\ and $c$, one can show that this
algorithm accepts a forrelated $\left\langle f,g\right\rangle $\ pair with
probability at least $2/3$, and accepts a random $\left\langle
f,g\right\rangle $\ pair with probability at least $1/3$. \ We omit the
details of the analysis, as they are tedious and not needed elsewhere in the paper.

In the next section, we will show that \textsc{Fourier Checking}\ has a
property called \textit{almost }$k$\textit{-wise independence}, which
immediately implies a lower bound of $\Omega\left(  \sqrt[4]{N}\right)
=\Omega\left(  2^{n/4}\right)  $\ on its classical query complexity (as well
as exponential lower bounds on its $\mathsf{MA}$, $\mathsf{BPP}_{\mathsf{path}%
}$, and $\mathsf{SZK}$\ query complexities). \ Indeed, we conjecture that
almost $k$-wise independence is enough to imply that \textsc{Fourier
Checking}\ is not in $\mathsf{PH}$. \ We discuss the status of that conjecture
in Section \ref{GLN}.

\subsection{Almost $k$-Wise Independence\label{INDEP}}

Let $Z=z_{1}\ldots z_{M}\in\left\{  -1,1\right\}  ^{M}$ be a string. \ Then a
\textit{literal} is a term of the form $\frac{1\pm z_{i}}{2}$, and a
$k$-\textit{term }is a product of $k$ literals (each involving a different
$z_{i}$), which is $1$ if the literals all take on prescribed values and $0$ otherwise.

Let $\mathcal{U}$\ be the uniform distribution over $\left\{  -1,1\right\}
^{M}$. \ The following definition will play a major role in this work.

\begin{definition}
A distribution $\mathcal{D}$ over $\left\{  -1,1\right\}  ^{M}$ is
$\varepsilon$\textit{-almost }$k$\textit{-wise independent} if for every
$k$-term $C$,%
\[
1-\varepsilon\leq\frac{\Pr_{\mathcal{D}}\left[  C\right]  }{\Pr_{\mathcal{U}%
}\left[  C\right]  }\leq1+\varepsilon.
\]
(Note that $\Pr_{\mathcal{U}}\left[  C\right]  $\ is just $2^{-k}$.)
\end{definition}

Now let $M=2^{n+1}=2N$, and let $\mathcal{F}$\ be the forrelated distribution
over pairs of Boolean functions $f,g:\left\{  0,1\right\}  ^{n}\rightarrow
\left\{  -1,1\right\}  $. \ That is, we sample $\left\langle f,g\right\rangle
\in\mathcal{F}$\ by first choosing a vector $v=\left(  v_{x}\right)
_{x\in\left\{  -1,1\right\}  ^{n}}\in\mathbb{R}^{N}$ of independent
$\mathcal{N}\left(  0,1\right)  $\ Gaussians, then setting $f\left(  x\right)
:=\operatorname*{sgn}\left(  v_{x}\right)  $\ for all $x$ and $g\left(
y\right)  :=\operatorname*{sgn}\left(  \widehat{v}_{y}\right)  $ for all $y$.

\begin{theorem}
\label{fcindep}For all $k\leq\sqrt[4]{N}$, the forrelated distribution
$\mathcal{F}$ is $O\left(  k^{2}/\sqrt{N}\right)  $-almost $k$-wise independent.
\end{theorem}

\begin{proof}
As a first step, we will prove an analogous statement for the real-valued
functions $F\left(  x\right)  :=v_{x}$\ and $G\left(  y\right)  :=\widehat
{v}_{y}$; then we will generalize to the discrete versions $f\left(  x\right)
$\ and $g\left(  y\right)  $. \ Let $\mathcal{U}^{\prime}$\ be the probability
measure over $\left\langle F,G\right\rangle $\ that corresponds to case (i) of
\textsc{Fourier Checking}: that is, we choose each $F\left(  x\right)  $\ and
$G\left(  y\right)  $\ independently from the Gaussian measure $\mathcal{N}%
\left(  0,1\right)  $. \ Let $\mathcal{F}^{\prime}$\ be the probability
measure over $\left\langle F,G\right\rangle $\ that corresponds to case (ii)
of \textsc{Fourier Checking}: that is, we choose each $F\left(  x\right)
$\ independently from $\mathcal{N}\left(  0,1\right)  $, then set $G\left(
y\right)  :=\widehat{F}\left(  y\right)  $\ where%
\[
\widehat{F}\left(  y\right)  =\frac{1}{\sqrt{N}}\sum_{x\in\left\{
0,1\right\}  ^{n}}\left(  -1\right)  ^{x\cdot y}F\left(  x\right)
\]
is the Fourier transform of $F$. \ Observe that since the Fourier transform is
unitary, $G$ has the same marginal distribution as $F$ under $\mathcal{F}%
^{\prime}$: namely, a product of independent $\mathcal{N}\left(  0,1\right)  $\ Gaussians.

Fix inputs $x_{1},\ldots,x_{K}\in\left\{  0,1\right\}  ^{n}$\ of $F$ and
$y_{1},\ldots,y_{L}\in\left\{  0,1\right\}  ^{n}$\ of $G$, for some $K,L\leq
N^{1/4}$. \ Then given constants $a_{1},\ldots,a_{K},b_{1},\ldots,b_{L}%
\in\mathbb{R}$, let $S$ be the set of all $\left\langle F,G\right\rangle
$\ that satisfy the $K+L$ equations
\begin{align}
F\left(  x_{i}\right)   &  =a_{i}~\text{for all }1\leq i\leq K\text{,}%
\label{fg}\\
G\left(  y_{j}\right)   &  =b_{j}~\text{for all }1\leq j\leq L\text{.}%
\nonumber
\end{align}
Clearly $S$ is a $\left(  2N-K-L\right)  $-dimensional affine subspace of
$\mathbb{R}^{2N}$. \ The \textit{measure} of $S$, under some probability
measure $\mu$ on $\mathbb{R}^{2N}$, is defined in the usual way as%
\[
\mu\left(  S\right)  :=\int_{\left\langle F,G\right\rangle \in S}\mu\left(
F,G\right)  d\left\langle F,G\right\rangle .
\]
Now let%
\[
\Delta_{S}:=a_{1}^{2}+\cdots+a_{K}^{2}+b_{1}^{2}+\cdots+b_{L}^{2}%
\]
be the squared distance between $S$ and the origin (that is, the minimum
squared $2$-norm of any point in $S$). \ Then by the spherical symmetry of the
Gaussian measure, it is not hard to see that $S$ has measure%
\[
\mathcal{U}^{\prime}\left(  S\right)  =\frac{e^{-\Delta_{S}/2}}{\sqrt{2\pi
}^{K+L}}%
\]
under $\mathcal{U}^{\prime}$. \ Our key claim is that%
\[
1-O\left(  \frac{\left(  K+L\right)  \Delta_{S}}{\sqrt{N}}\right)  \leq
\frac{\mathcal{F}^{\prime}\left(  S\right)  }{\mathcal{U}^{\prime}\left(
S\right)  }\leq1+O\left(  \frac{\left(  K+L\right)  \Delta_{S}}{\sqrt{N}%
}\right)  .
\]

To prove this claim: recall that the probability measure over $F$ induced by
$\mathcal{F}^{\prime}$ is just a spherical Gaussian $\mathcal{G}$\ on
$\mathbb{R}^{N}$, and that $G=\widehat{F}$ uniquely determines $F$ and vice
versa. \ So consider the $\left(  N-K-L\right)  $-dimensional affine subspace
$T$ of $\mathbb{R}^{N}$\ defined by the $K+L$\ equations%
\begin{align*}
F\left(  x_{i}\right)   &  =a_{i}~\text{for all }1\leq i\leq K\text{,}\\
\widehat{F}\left(  y_{j}\right)   &  =b_{j}~\text{for all }1\leq j\leq
L\text{.}%
\end{align*}
Then $\mathcal{F}^{\prime}\left(  S\right)  =\mathcal{G}\left(  T\right)  $:
that is, to compute\ how much measure $\mathcal{F}^{\prime}$\ assigns to $S$,
it suffices to compute how much measure $\mathcal{G}$\ assigns to $T$. \ We
have%
\[
\mathcal{G}\left(  T\right)  =\frac{e^{-\Delta_{T}/2}}{\sqrt{2\pi}^{K+L}},
\]
where $\Delta_{T}$\ is the squared Euclidean distance between $T$\ and the
origin. \ Thus, our problem reduces to minimizing%
\[
\Delta_{F}:=\sum_{x\in\left\{  0,1\right\}  ^{n}}F\left(  x\right)  ^{2}%
\]
over all $F\in T$. \ By a standard fact about quadratic optimization, the
minimal $F\in T$ will have the form%
\[
F\left(  x\right)  =\alpha_{1}E_{1}\left(  x\right)  +\cdots+\alpha_{K}%
E_{K}\left(  x\right)  +\beta_{1}\chi_{1}\left(  x\right)  +\cdots+\beta
_{L}\chi_{L}\left(  x\right)
\]
where%
\[
E_{i}\left(  x\right)  :=\left\{
\begin{array}
[c]{cc}%
1 & \text{if }x=x_{i}\\
0 & \text{otherwise}%
\end{array}
\right.
\]
is an indicator function, and%
\[
\chi_{j}\left(  x\right)  :=\frac{\left(  -1\right)  ^{x\cdot y_{j}}}{\sqrt
{N}}%
\]
is the $y_{j}^{th}$\ Fourier character evaluated at $x$. \ Furthermore, the
coefficients $\left\{  \alpha_{i}\right\}  _{i\in\left[  K\right]  },\left\{
\beta_{j}\right\}  _{j\in\left[  L\right]  }$ can be obtained by solving the
linear system%
\[
\underset{A}{\underbrace{\left(
\begin{array}
[c]{cccccc}%
1 & 0 & 0 & \pm1/\sqrt{N} & \cdots & \pm1/\sqrt{N}\\
0 & \ddots & 0 & \vdots & \ddots & \vdots\\
0 & 0 & 1 & \pm1/\sqrt{N} & \cdots & \pm1/\sqrt{N}\\
\pm1/\sqrt{N} & \cdots & \pm1/\sqrt{N} & 1 & 0 & 0\\
\vdots & \ddots & \vdots & 0 & \ddots & 0\\
\pm1/\sqrt{N} & \cdots & \pm1/\sqrt{N} & 0 & 0 & 1
\end{array}
\right)  }}\underset{u}{\underbrace{\left(
\begin{array}
[c]{c}%
\alpha_{1}\\
\vdots\\
\alpha_{K}\\
\beta_{1}\\
\vdots\\
\beta_{L}%
\end{array}
\right)  }}=\underset{w}{\underbrace{\left(
\begin{array}
[c]{c}%
a_{1}\\
\vdots\\
a_{K}\\
b_{1}\\
\vdots\\
b_{L}%
\end{array}
\right)  }}%
\]
Here $A$ is simply a matrix of covariances: the top left block records the
inner product between each $E_{i}$\ and $E_{j}$ (and hence is a $K\times
K$\ identity matrix), the bottom right block records the inner product between
each $\chi_{i}$\ and $\chi_{j}$ (and hence is an $L\times L$\ identity
matrix), and the remaining two blocks of size $K\times L$\ record the inner
product between each $E_{i}$\ and $\chi_{j}$.

Thus, to get the vector of coefficients $u\in\mathbb{R}^{K+L}$, we simply need
to calculate $A^{-1}w$. \ Define $B:=I-A$. \ Then by Taylor series expansion,%
\[
A^{-1}=\left(  I-B\right)  ^{-1}=I+B+B^{2}+B^{3}+\cdots
\]
Notice that every entry of $B$ is at most $1/\sqrt{N}$\ in absolute value.
\ This means that, for all positive integers $t$, every entry of $B^{t}$\ is
at most%
\[
\frac{\left(  K+L\right)  ^{t-1}}{N^{t/2}}%
\]
in absolute value. \ Since $K+L\ll\sqrt{N}$, this in turn means that every
entry of $I-A^{-1}$\ has absolute value $O\left(  1/\sqrt{N}\right)  $. \ So
$A^{-1}$\ is exponentially close to the identity matrix. \ Hence, when we
compute the vector $u=A^{-1}w$, we find that%
\begin{align*}
\alpha_{i} &  =a_{i}+\varepsilon_{i}~\text{for all }1\leq i\leq K\text{,}\\
\beta_{j} &  =b_{j}+\delta_{j}~\text{for all }1\leq j\leq L\text{,}%
\end{align*}
for some small error terms $\varepsilon_{i}$\ and $\delta_{j}$.
\ Specifically, each $\varepsilon_{i}$\ and $\delta_{j}$\ is the inner product
of $w$, a $\left(  K+L\right)  $-dimensional vector of length $\sqrt
{\Delta_{S}}$, with a vector every entry of which has absolute value $O\left(
1/\sqrt{N}\right)  $. \ By Cauchy-Schwarz, this implies that%
\[
\left\vert \varepsilon_{i}\right\vert ,\left\vert \delta_{j}\right\vert
=O\left(  \frac{\sqrt{\left(  K+L\right)  \Delta_{S}}}{\sqrt{N}}\right)
\]
for all $i,j$. \ So%
\begin{align*}
\Delta_{T} &  =\min_{F\in T}\sum_{x\in\left\{  0,1\right\}  ^{n}}F\left(
x\right)  ^{2}\\
&  =\sum_{i=1}^{K}\alpha_{i}^{2}+\sum_{j=1}^{L}\beta_{j}^{2}+2\sum_{i=1}%
^{K}\sum_{j=1}^{L}\frac{\alpha_{i}\beta_{j}}{\sqrt{N}}\\
&  =\sum_{i=1}^{K}\left(  a_{i}+\varepsilon_{i}\right)  ^{2}+\sum_{j=1}%
^{L}\left(  b_{j}+\delta_{j}\right)  ^{2}+2\sum_{i=1}^{K}\sum_{j=1}^{L}%
\frac{\left(  a_{i}+\varepsilon_{i}\right)  \left(  b_{j}+\delta_{j}\right)
}{\sqrt{N}}\\
&  =\Delta_{S}\pm O\left(  \frac{\left(  K+L\right)  \Delta_{S}}{\sqrt{N}%
}+\frac{\left(  K+L\right)  ^{2}\Delta_{S}}{N}+\frac{\left(  K+L\right)
^{3}\Delta_{S}}{N^{3/2}}\right)  \\
&  =\Delta_{S}\left(  1\pm O\left(  \frac{K+L}{\sqrt{N}}\right)  \right)  ,
\end{align*}
where the fourth line made repeated use of Cauchy-Schwarz, and the fifth line
used the fact that $K+L\ll\sqrt{N}$. \ Hence%
\begin{align*}
\frac{\mathcal{F}^{\prime}\left(  S\right)  }{\mathcal{U}^{\prime}\left(
S\right)  } &  =\frac{e^{-\Delta_{T}/2}/\sqrt{2\pi}^{K+L}}{e^{-\Delta_{S}%
/2}/\sqrt{2\pi}^{K+L}}\\
&  =\exp\left(  \frac{\Delta_{S}-\Delta_{T}}{2}\right)  \\
&  =\exp\left(  \pm O\left(  \frac{\left(  K+L\right)  \Delta_{S}}{\sqrt{N}%
}\right)  \right)  \\
&  =1\pm O\left(  \frac{\left(  K+L\right)  \Delta_{S}}{\sqrt{N}}\right)
\end{align*}
which proves the claim.

To prove the theorem, we now need to generalize to the discrete functions
$f$\ and $g$. \ Here we are given a term $C$ that is a conjunction of $K+L$
inequalities: $K$ of the form $F\left(  x_{i}\right)  \leq0$ or $F\left(
x_{i}\right)  \geq0$, and $L$ of the form $G\left(  y_{j}\right)  \leq0$ or
$G\left(  y_{j}\right)  \geq0$. \ If we fix $x_{1},\ldots,x_{K}$ and
$y_{1},\ldots,y_{L}$, we can think of $C$ as just a convex region of
$\mathbb{R}^{K+L}$. \ Then given an affine subspace $S$ as defined by equation
(\ref{fg}), we will (abusing notation) write $S\in C$\ if the vector $\left(
\alpha_{1},\ldots,\alpha_{K},\beta_{1},\ldots,\beta_{L}\right)  $\ is in $C$:
that is, if $S$ is compatible with the $K+L$\ inequalities that define $C$.
\ We need to show that the ratio $\Pr_{\mathcal{F}}\left[  C\right]
/\Pr_{\mathcal{U}}\left[  C\right]  $\ is close to $1$. \ We can do so using
the previous result, as follows:%
\begin{align*}
\frac{\Pr_{\mathcal{F}}\left[  C\right]  }{\Pr_{\mathcal{U}}\left[  C\right]
}  &  =\frac{\int_{S\in C}\mathcal{F}^{\prime}\left(  S\right)  dS}{\int_{S\in
C}\mathcal{U}^{\prime}\left(  S\right)  dS}\\
&  =\frac{\int_{S\in C}\mathcal{U}^{\prime}\left(  S\right)  \left[  1\pm
O\left(  \frac{\left(  K+L\right)  \Delta_{S}}{\sqrt{N}}\right)  \right]
dS}{\int_{S\in C}\mathcal{U}^{\prime}\left(  S\right)  dS}\\
&  =\frac{\int_{S\in C}\left[  e^{-\Delta_{S}/2}/\sqrt{2\pi}^{K+L}\right]
\left[  1\pm O\left(  \frac{\left(  K+L\right)  \Delta_{S}}{\sqrt{N}}\right)
\right]  dS}{\int_{S\in C}\left[  e^{-\Delta_{S}/2}/\sqrt{2\pi}^{K+L}\right]
dS}\\
&  =\frac{\left(  1/2\right)  ^{K+L}\pm O\left(  \int_{S\in C}\left[
e^{-\Delta_{S}/2}/\sqrt{2\pi}^{K+L}\right]  \frac{\left(  K+L\right)
\Delta_{S}}{\sqrt{N}}dS\right)  }{\left(  1/2\right)  ^{K+L}}\\
&  =1\pm\frac{2^{K+L}\left(  K+L\right)  }{\sqrt{N}}O\left(  \int_{S\in
C}\frac{e^{-\Delta_{S}/2}}{\sqrt{2\pi}^{K+L}}\Delta_{S}dS\right) \\
&  =1\pm\frac{K+L}{\sqrt{N}}O\left(  \int_{S}\frac{e^{-\Delta_{S}/2}}%
{\sqrt{2\pi}^{K+L}}\Delta_{S}dS\right) \\
&  =1\pm O\left(  \frac{\left(  K+L\right)  ^{2}}{\sqrt{N}}\right)  .
\end{align*}
Setting $k:=K+L$, this completes the proof.
\end{proof}

\subsection{Oracle Separation Results\label{OSEP}}

The following lemma shows that \textit{any} almost $k$-wise independent
distribution is indistinguishable from the uniform distribution by
$\mathsf{BPP}_{\mathsf{path}}$\ or $\mathsf{SZK}$\ machines.

\begin{lemma}
\label{bpppathszk}Suppose a probability distribution $\mathcal{D}$ over oracle
strings is $1/t\left(  n\right)  $-almost\ $\operatorname*{poly}\left(
n\right)  $-wise independent,\ for some superpolynomial function $t$. \ Then
no $\mathsf{BPP}_{\mathsf{path}}$\ machine or $\mathsf{SZK}$\ protocol can
distinguish $\mathcal{D}$\ from the uniform distribution $\mathcal{U}$\ with
non-negligible bias.
\end{lemma}

\begin{proof}
Let $M$ be a $\mathsf{BPP}_{\mathsf{path}}$\ machine, and let $p_{\mathcal{D}%
}$\ be the probability that $M$ accepts an oracle string drawn from
distribution $\mathcal{D}$. \ Then $p_{\mathcal{D}}$\ can be written as
$a_{\mathcal{D}}/s_{\mathcal{D}}$, where $s_{\mathcal{D}}$\ is the fraction of
$M$'s computation paths that are postselected, and $a_{\mathcal{D}}$\ is the
fraction of $M$'s paths that are both postselected and accepting. \ Since each
computation path can examine at most $\operatorname*{poly}\left(  n\right)
$\ bits\ and $\mathcal{D}$ is $1/t\left(  n\right)  $%
-almost\ $\operatorname*{poly}\left(  n\right)  $-wise independent, we have%
\[
1-\frac{1}{t\left(  n\right)  }\leq\frac{a_{\mathcal{D}}}{a_{\mathcal{U}}}%
\leq1+\frac{1}{t\left(  n\right)  }~\ \ \text{and \ }1-\frac{1}{t\left(
n\right)  }\leq\frac{s_{\mathcal{D}}}{s_{\mathcal{U}}}\leq1+\frac{1}{t\left(
n\right)  }.
\]
Hence%
\[
\left(  1-\frac{1}{t\left(  n\right)  }\right)  ^{2}\leq\frac{a_{\mathcal{D}%
}/s_{\mathcal{D}}}{a_{\mathcal{U}}/s_{\mathcal{U}}}\leq\left(  1+\frac
{1}{t\left(  n\right)  }\right)  ^{2}.
\]

Now let $P$ be an $\mathsf{SZK}$\ protocol. \ Then by a result of Sahai and
Vadhan \cite{sv}, there exist polynomial-time samplable distributions $A$\ and
$A^{\prime}$\ such that if $P$ accepts, then $\left\Vert A-A^{\prime
}\right\Vert \leq1/3$, while if $P$ rejects, then $\left\Vert A-A^{\prime
}\right\Vert \geq2/3$. \ But since each computation path can examine at most
$\operatorname*{poly}\left(  n\right)  $\ oracle bits\ and $\mathcal{D}$ is
$1/t\left(  n\right)  $-almost\ $\operatorname*{poly}\left(  n\right)  $-wise
independent, we have $\left\Vert A_{\mathcal{D}}-A_{\mathcal{U}}\right\Vert
\leq1/t\left(  n\right)  $\ and $\left\Vert A_{\mathcal{D}}^{\prime
}-A_{\mathcal{U}}^{\prime}\right\Vert \leq1/t\left(  n\right)  $, where the
subscript denotes the distribution from which the oracle string was drawn.
\ Hence%
\[
\left\vert \left\Vert A_{\mathcal{D}}-A_{\mathcal{D}}^{\prime}\right\Vert
-\left\Vert A_{\mathcal{U}}-A_{\mathcal{U}}^{\prime}\right\Vert \right\vert
\leq\left\Vert A_{\mathcal{D}}-A_{\mathcal{U}}\right\Vert +\left\Vert
A_{\mathcal{D}}^{\prime}-A_{\mathcal{U}}^{\prime}\right\Vert \leq\frac
{2}{t\left(  n\right)  }%
\]
and no $\mathsf{SZK}$\ protocol\ exists.
\end{proof}

We now combine Lemma \ref{bpppathszk} and Theorem \ref{fcindep} with standard
diagonalization tricks, to obtain an oracle relative to which $\mathsf{BQP}%
\not \subset \mathsf{BPP}_{\mathsf{path}}$\ and $\mathsf{BQP}\not \subset
\mathsf{SZK}$.

\begin{theorem}
\label{bqpbpppath}There exists an oracle $A$\ relative to which $\mathsf{BQP}%
^{A}\not \subset \mathsf{BPP}_{\mathsf{path}}^{A}$ and $\mathsf{BQP}%
^{A}\not \subset \mathsf{SZK}^{A}$.
\end{theorem}

\begin{proof}
The oracle $A$ will encode the truth tables of Boolean functions $f_{1}%
,f_{2},\ldots$\ and $g_{1},g_{2},\ldots$, where $f_{n},g_{n}:\left\{
0,1\right\}  ^{n}\rightarrow\left\{  -1,1\right\}  $\ are on $n$\ variables
each.\ \ For each $n$, with $1/2$\ probability\ we draw $\left\langle
f_{n},g_{n}\right\rangle $\ from the uniform distribution $\mathcal{U}$, and
with $1/2$\ probability\ we draw $\left\langle f_{n},g_{n}\right\rangle
$\ from the forrelated distribution $\mathcal{F}$. \ Let $L$\ be the unary
language consisting of all $0^{n}$\ for which $\left\langle f_{n}%
,g_{n}\right\rangle $\ was drawn from $\mathcal{F}$.

By Theorem \ref{fcinbqp}, there exists a $\mathsf{BQP}^{A}$\ machine $M$\ that
decides $L$ on all but finitely many values of $n$, with probability $1$ over
$A$. \ Since we can simply hardwire the values of $n$ on which $M$ fails, it
follows that $L\in\mathsf{BQP}^{A}$ with probability $1$ over $A$.

On the other hand, we showed in Theorem \ref{fcindep} that $\mathcal{F}$\ is
$O\left(  p\left(  n\right)  ^{2}/2^{n/2}\right)  $-almost $p\left(  n\right)
$-wise\ independent for all polynomials $p$. \ Hence, by Lemma
\ref{bpppathszk}, no $\mathsf{BPP}_{\mathsf{path}}$\ machine can distinguish
$\mathcal{F}$\ from $\mathcal{U}$\ with non-negligible bias. \ Let
$E_{n}\left(  M\right)  $ be the event that the $\mathsf{BPP}_{\mathsf{path}%
}^{A}$\ machine $M$\ correctly decides whether $0^{n}\in L$. \ Then%
\[
\Pr_{A}\left[  E_{n}\left(  M\right)  \right]  \leq\frac{1}{2}+o\left(
1\right)  ,
\]
and moreover this is true even conditioning on $E_{1}\left(  M\right)
,\ldots,E_{n-1}\left(  M\right)  $. \ So as in the standard random oracle
argument of Bennett and Gill \cite{bg}, for every fixed $M$ we have%
\[
\Pr_{A}\left[  E_{1}\left(  M\right)  \wedge E_{2}\left(  M\right)
\wedge\cdots\right]  =0.
\]
So by the union bound,%
\[
\Pr_{A}\left[  \exists M:E_{1}\left(  M\right)  \wedge E_{2}\left(  M\right)
\wedge\cdots\right]  =0
\]
as well. \ It follows that $\mathsf{BQP}^{A}\not \subset \mathsf{BPP}%
_{\mathsf{path}}^{A}$\ with probability $1$ over $A$. \ By exactly the same
argument, we also get $\mathsf{BQP}^{A}\not \subset \mathsf{SZK}^{A}$ with
probability $1$ over $A$.
\end{proof}

Since $\mathsf{BPP}\subseteq\mathsf{MA}\subseteq\mathsf{BPP}_{\mathsf{path}}$,
Theorem \ref{bqpbpppath} supersedes the previous results that there exist
oracles $A$ relative to which $\mathsf{BPP}^{A}\neq\mathsf{BQP}^{A}$ \cite{bv}
and $\mathsf{BQP}^{A}\not \subset \mathsf{MA}^{A}$ \cite{watrous}.

\section{The Generalized Linial-Nisan Conjecture\label{GLN}}

In 1990, Linial and Nisan \cite{ln} famously conjectured that
\textquotedblleft polylogarithmic independence fools $\mathsf{AC}^{0}%
$\textquotedblright---or loosely speaking, that every probability distribution
$\mathcal{D}$\ over $n$-bit strings that is uniform\ on all small subsets of
bits, is \textit{indistinguishable} from the uniform distribution by
polynomial-size, constant-depth circuits.\ \ We now state a variant of the
Linial-Nisan Conjecture,\ not with the best possible parameters but with
weaker, easier-to-understand parameters that suffice for our application.

\begin{conjecture}
[Linial-Nisan Conjecture]\label{ln}Let $\mathcal{D}$\ be an $n^{\Omega\left(
1\right)  }$-wise independent distribution over $\left\{  0,1\right\}  ^{n}$,
and let $f:\left\{  0,1\right\}  ^{n}\rightarrow\left\{  0,1\right\}  $\ be
computed by an $\mathsf{AC}^{0}$ circuit of size $2^{n^{o\left(  1\right)  }}%
$\ and depth $O\left(  1\right)  $. \ Then%
\[
\left\vert \Pr_{x\sim\mathcal{D}}\left[  f\left(  x\right)  \right]
-\Pr_{x\sim\mathcal{U}}\left[  f\left(  x\right)  \right]  \right\vert
=o\left(  1\right)  .
\]

\end{conjecture}

After seventeen years of almost no progress, in 2007 Bazzi \cite{bazzi}
finally proved Conjecture \ref{ln}\ for the special case of depth-$2$ circuits
(also called DNF formulas). \ Bazzi's proof was about $50$ pages, but it was
dramatically simplified a year later, when Razborov \cite{razborov:bazzi}
discovered a $3$-page proof. \ Then in 2009, Braverman \cite{braverman} gave a
breakthrough proof of the full Linial-Nisan Conjecture.

\begin{theorem}
[Braverman's Theorem \cite{braverman}]\label{bravermanthm}Let $f:\left\{
0,1\right\}  ^{n}\rightarrow\left\{  0,1\right\}  $\ be computed by an
$\mathsf{AC}^{0}$ circuit of size $S$\ and depth $d$, and let $\mathcal{D}%
$\ be a $\left(  \log\frac{S}{\varepsilon}\right)  ^{7d^{2}}$-wise independent
distribution over $\left\{  0,1\right\}  ^{n}$.\ \ Then for all sufficiently
large $S$,%
\[
\left\vert \Pr_{x\sim\mathcal{D}}\left[  f\left(  x\right)  \right]
-\Pr_{x\sim\mathcal{U}}\left[  f\left(  x\right)  \right]  \right\vert
\leq\varepsilon.
\]

\end{theorem}

We conjecture a modest-seeming extension of Braverman's Theorem, which says
(informally) that \textit{almost} $k$-wise independent distributions fool
$\mathsf{AC}^{0}$ as well.

\begin{conjecture}
[Generalized Linial-Nisan or GLN Conjecture]\label{gln}Let $\mathcal{D}$\ be a
$1/n^{\Omega\left(  1\right)  }$-almost $n^{\Omega\left(  1\right)  }$-wise
independent distribution over $\left\{  0,1\right\}  ^{n}$, and let
$f:\left\{  0,1\right\}  ^{n}\rightarrow\left\{  0,1\right\}  $\ be computed
by an $\mathsf{AC}^{0}$ circuit of size $2^{n^{o\left(  1\right)  }}$\ and
depth $O\left(  1\right)  $. \ Then%
\[
\left\vert \Pr_{x\sim\mathcal{D}}\left[  f\left(  x\right)  \right]
-\Pr_{x\sim\mathcal{U}}\left[  f\left(  x\right)  \right]  \right\vert
=o\left(  1\right)  .
\]

\end{conjecture}

By the usual correspondence between $\mathsf{AC}^{0}$\ and $\mathsf{PH}$, the
GLN Conjecture immediately implies the following counterpart of Lemma
\ref{bpppathszk}.

\begin{quotation}
\noindent\textit{Suppose a probability distribution }$\mathcal{D}$\textit{
over oracle strings is }$1/t\left(  n\right)  $\textit{-almost\ }%
$\operatorname*{poly}\left(  n\right)  $\textit{-wise independent,\ for some
superpolynomial function }$t$\textit{. \ Then no }$\mathsf{PH}$%
\textit{\ machine can distinguish }$\mathcal{D}$\textit{\ from the uniform
distribution }$\mathcal{U}$\textit{\ with non-negligible bias.}
\end{quotation}

And thus we get the following implication:

\begin{theorem}
\label{glnimp}Assuming the GLN Conjecture, there exists an oracle $A$ relative
to which $\mathsf{BQP}^{A}\not \subset \mathsf{PH}^{A}$.
\end{theorem}

\begin{proof}
The proof is the same as that of Theorem \ref{bqpbpppath}; the only difference
is that the GLN Conjecture\ now plays the role of Lemma \ref{bpppathszk}.
\end{proof}

Likewise:

\begin{theorem}
\label{glnam}Assuming the GLN Conjecture for the special case of depth-$2$
circuits (i.e., DNF formulas), there exists an oracle $A$ relative to which
$\mathsf{BQP}^{A}\not \subset \mathsf{AM}^{A}$.
\end{theorem}

\begin{proof}
Just like in Theorem \ref{bqpbpppath}, define an oracle $A$ and an associated
language $L$\ using the \textsc{Fourier Checking}\ problem. \ Then
$L\in\mathsf{BQP}^{A}$, with probability $1$ over the choices made in
constructing $A$. \ On the other hand, suppose $L\in\mathsf{AM}^{A}$ with
probability $1$\ over $A$. \ Then we claim that \textsc{Fourier Checking} can
also be solved by a family of DNF formulas $\left\{  \varphi_{n}\right\}
_{n\geq1}$ of size $2^{\operatorname*{poly}\left(  n\right)  }$:%
\[
\left\vert \Pr_{\left\langle f,g\right\rangle \sim\mathcal{F}}\left[
\varphi_{n}\left(  f,g\right)  \right]  -\Pr_{\left\langle f,g\right\rangle
\sim\mathcal{U}}\left[  \varphi_{n}\left(  f,g\right)  \right]  \right\vert
=\Omega\left(  1\right)  .
\]
But since $\mathcal{F}$\ is $O\left(  k^{2}/2^{n/2}\right)  $-almost $k$-wise
independent (by Theorem \ref{fcindep}), such a family $\varphi_{n}$\ would
violate the depth-$2$ case of the GLN Conjecture.

We now prove the claim. For simplicity, fix an input length $n$, and let $A$
refer to a single instance $\left\langle f,g\right\rangle $\ of
\textsc{Fourier Checking}.\footnote{It is straightforward to generalize to the
case where Arthur can query other instances, besides the one he is trying to
solve.} \ Let $P$ be an $\mathsf{AM}$\ protocol that successfully
distinguishes the forrelated distribution $\mathcal{F}$\ over $\left\langle
f,g\right\rangle $\ pairs from the uniform distribution $\mathcal{U}$. \ We
can assume without loss of generality that $P$ is \textit{public-coin}
\cite{gs}. \ In other words, Arthur first sends a random challenge
$r\in\left\{  0,1\right\}  ^{\operatorname*{poly}\left(  n\right)  }$\ to
Merlin,\ then Merlin responds with a witness $w\in\left\{  0,1\right\}
^{\operatorname*{poly}\left(  n\right)  }$, then Arthur runs a deterministic
polynomial-time verification procedure $V^{A}\left(  r,w\right)  $ to decide
whether to accept. \ By the assumption that $P$ succeeds,%
\[
\left\vert \Pr_{A\sim\mathcal{D},r}\left[  \exists w:V^{A}\left(  r,w\right)
\right]  -\Pr_{A\sim\mathcal{D},r}\left[  \exists w:V^{A}\left(  r,w\right)
\right]  \right\vert =\Omega\left(  1\right)  .
\]
So by Yao's principle, there exists a \textit{fixed} challenge $r^{\ast}%
$\ such that%
\[
\left\vert \Pr_{A\sim\mathcal{D}}\left[  \exists w:V^{A}\left(  r^{\ast
},w\right)  \right]  -\Pr_{A\sim\mathcal{D}}\left[  \exists w:V^{A}\left(
r^{\ast},w\right)  \right]  \right\vert =\Omega\left(  1\right)  .
\]
Now let $Q_{A,w}$\ be the set of all queries that $V^{A}\left(  r^{\ast
},w\right)  $ makes to $A$, and let $C_{A,w}\left(  A^{\prime}\right)  $\ be a
term (i.e., a conjunction of $1$'s and $0$'s)\ that returns TRUE if and only
if $A^{\prime}$\ agrees with $A$ on all queries in $Q_{A,w}$. \ Then we can
assume without loss of generality that $C_{w}:=C_{A,w}$\ depends only on $w$,
not on $A$---since Merlin can simply \textit{tell} Arthur what queries $V$ is
going to make and what their outcomes will be, and Arthur can reject if Merlin
is lying. \ Let $W$ be the set of all witnesses $w$\ such that Arthur accepts
if $C_{w}\left(  A\right)  $ returns TRUE. \ Consider the DNF formula%
\[
\varphi\left(  A\right)  :=%
%TCIMACRO{\dbigvee \limits_{w\in W}}%
%BeginExpansion
{\displaystyle\bigvee\limits_{w\in W}}
%EndExpansion
C_{w}\left(  A\right)  ,
\]
which expresses that there exists a $w$ causing $V^{A}\left(  r^{\ast
},w\right)  $\ to accept. \ Then $\varphi$\ contains at most
$2^{\operatorname*{poly}\left(  n\right)  }$\ terms with $\operatorname*{poly}%
\left(  n\right)  $\ literals each, and%
\[
\left\vert \Pr_{A\sim\mathcal{D}}\left[  \varphi\left(  A\right)  \right]
-\Pr_{A\sim\mathcal{D}}\left[  \varphi\left(  A\right)  \right]  \right\vert
=\Omega\left(  1\right)  .
\]

\end{proof}

As a side note, it is conceivable that one could prove%
\[
\Pr_{x\sim\mathcal{D}}\left[  \varphi\left(  x\right)  \right]  -\Pr
_{x\sim\mathcal{U}}\left[  \varphi\left(  x\right)  \right]  =o\left(
1\right)
\]
for every almost $k$-wise independent distribution $\mathcal{D}$\ and small
\textit{CNF} formula $\varphi$, without getting the same result for
\textit{DNF} formulas (or vice versa). \ However, since $\mathsf{BQP}$\ is
closed under complement, even such an asymmetric result would imply an oracle
$A$ relative to which $\mathsf{BQP}^{A}\not \subset \mathsf{AM}^{A}$.

If the GLN Conjecture\ holds, then we can also \textquotedblleft scale down by
an exponential,\textquotedblright\ to obtain an \textit{unrelativized}
decision problem that is solvable in quantum logarithmic time but not in
$\mathsf{AC}^{0}$.

\begin{theorem}
\label{glnimpbqlog}Assuming the GLN Conjecture, there exists a promise problem
in $\mathsf{BQLOGTIME}$\ that is not in $\mathsf{AC}^{0}$.
\end{theorem}

\begin{proof}
In our promise problem $\Pi=\left(  \Pi_{\operatorname*{YES}},\Pi
_{\operatorname*{NO}}\right)  $, the inputs (of size $M=2^{n+1}$) will encode
pairs of Boolean functions $f,g:\left\{  0,1\right\}  ^{n}\rightarrow\left\{
-1,1\right\}  $, such that%
\[
p\left(  f,g\right)  :=\frac{1}{N^{3}}\left(  \sum_{x,y\in\left\{
0,1\right\}  ^{n}}f\left(  x\right)  \left(  -1\right)  ^{x\cdot y}g\left(
y\right)  \right)  ^{2}%
\]
is either at least $0.05$\ or at most $0.01$. \ The problem is to accept in
the former case and reject in the latter case.

Using the algorithm \texttt{FC-ALG}\ from Section \ref{FCALG}, it is immediate
that $\Pi\in\mathsf{BQLOGTIME}$. \ On the other hand, suppose $\Pi
\in\mathsf{AC}^{0}$. \ Then we get a nonuniform circuit family $\left\{
C_{n}\right\}  _{n}$, of depth $O\left(  1\right)  $\ and size
$\operatorname*{poly}\left(  M\right)  =2^{O\left(  n\right)  }$,\ that solves
\textsc{Fourier Checking}\ on all pairs $\left\langle f,g\right\rangle $\ such
that (i) $p\left(  f,g\right)  \leq0.01$\ or (ii) $p\left(  f,g\right)
\geq0.05$. \ By Corollary \ref{pfgcor}, the class (i) includes the
overwhelming majority of $\left\langle f,g\right\rangle $'s drawn from the
uniform distribution $\mathcal{U}$, while the class (ii) includes a constant
fraction of $\left\langle f,g\right\rangle $'s drawn from the forrelated
distribution $\mathcal{F}$. \ Therefore, we actually obtain an $\mathsf{AC}%
^{0}$ circuit family that distinguishes $\mathcal{U}$\ from $\mathcal{F}%
$\ with constant bias. \ But this contradicts Theorem \ref{fcindep} together
with the GLN Conjecture.
\end{proof}

\subsection{Low-Fat Polynomials\label{LOWFAT}}

Given that the GLN Conjecture would have such remarkable implications for
quantum complexity theory, the question arises of how we can go about proving
it. \ As we are indebted to Louay Bazzi for pointing out to us, the GLN
Conjecture\ is \textit{equivalent} to the following conjecture, about
approximating $\mathsf{AC}^{0}$\ functions by low-degree polynomials.

\begin{conjecture}
[Low-Fat Sandwich Conjecture]\label{sandwich}For every function $f:\left\{
0,1\right\}  ^{n}\rightarrow\left\{  0,1\right\}  $ computable by an
$\mathsf{AC}^{0}$\ circuit, there exist polynomials\ $p_{\ell},p_{u}%
:\mathbb{R}^{n}\rightarrow\mathbb{R}$ of degree $k=n^{o\left(  1\right)  }$
that satisfy the following three conditions.

\begin{enumerate}
\item[(i)] \textbf{Sandwiching:} $p_{\ell}\left(  x\right)  \leq f\left(
x\right)  \leq p_{u}\left(  x\right)  $ for all $x\in\left\{  0,1\right\}
^{n}$.

\item[(ii)] $L_{1}$\textbf{-Approximation:}$\ \operatorname*{E}_{x\sim
\mathcal{U}}\left[  p_{u}\left(  x\right)  -p_{\ell}\left(  x\right)  \right]
=o\left(  1\right)  $.

\item[(iii)] \textbf{Low-Fat:} $p_{\ell}\left(  x\right)  $ and $p_{u}\left(
x\right)  $\ can be written as linear combinations of terms, $p_{\ell}\left(
x\right)  =\sum_{C}\alpha_{C}C\left(  x\right)  $ and $p_{u}\left(  x\right)
=\sum_{C}\beta_{C}C\left(  x\right)  $\ respectively, such that $\sum
_{C}\left\vert \alpha_{C}\right\vert 2^{-\left\vert C\right\vert }=n^{o\left(
1\right)  }$ and $\sum_{C}\left\vert \beta_{C}\right\vert 2^{-\left\vert
C\right\vert }=n^{o\left(  1\right)  }$. \ (Here a term is a product of
literals of the form $x_{i}$\ and $1-x_{i}$.)
\end{enumerate}
\end{conjecture}

If we take out condition (iii), then Conjecture \ref{sandwich}\ becomes
equivalent to the \textit{original} Linial-Nisan Conjecture (see Bazzi
\cite{bazzi} for a proof). \ And indeed, all progress so far on
\textquotedblleft Linial-Nisan problems\textquotedblright\ has crucially
relied on this connection with polynomials. \ Bazzi \cite{bazzi}\ and Razborov
\cite{razborov:bazzi}\ proved the depth-$2$ case of the LN Conjecture by
constructing low-degree, approximating, sandwiching polynomials for every DNF,
while Braverman \cite{braverman} proved the full LN Conjecture by constructing
such polynomials for every $\mathsf{AC}^{0}$\ circuit.\footnote{Strictly
speaking, Braverman constructed approximating polynomials with slightly
different (though still sufficient) properties.\ \ We know from Bazzi
\cite{bazzi}\ that it must be possible to get sandwiching polynomials as
well.} \ Given this history, proving Conjecture \ref{sandwich}\ would seem
like the \textquotedblleft obvious\textquotedblright\ approach to proving the
GLN Conjecture.

Below we prove one direction of the equivalence: that to prove the GLN
Conjecture, it suffices to construct low-fat sandwiching polynomials for every
$\mathsf{AC}^{0}$\ circuit. \ The other direction---that the GLN Conjecture
implies Conjecture \ref{sandwich}, and hence, there is no loss of generality
in working with polynomials instead of probability distributions---follows
from a linear programming duality calculation that we omit.

\begin{theorem}
\label{sandimpgln}The Low-Fat Sandwich Conjecture\ implies the GLN Conjecture.
\end{theorem}

\begin{proof}
Given an $\mathsf{AC}^{0}$\ function $f$, let $p_{\ell},p_{u}$\ be the low-fat
sandwiching polynomials of degree $k$ that are guaranteed by Conjecture
\ref{sandwich}. \ Also, let $\mathcal{D}$\ be an $\varepsilon$-almost $k$-wise
independent distribution over $\left\{  0,1\right\}  ^{n}$, for some
$\varepsilon=1/n^{\Omega\left(  1\right)  }$. \ Then%
\begin{align*}
\Pr_{x\sim\mathcal{D}}\left[  f\left(  x\right)  \right]  -\Pr_{x\sim
\mathcal{U}}\left[  f\left(  x\right)  \right]   &  \leq\operatorname*{E}%
_{\mathcal{D}}\left[  p_{u}\right]  -\operatorname*{E}_{\mathcal{U}}\left[
p_{\ell}\right] \\
&  =\sum_{C}\beta_{C}\operatorname*{E}_{\mathcal{D}}\left[  C\right]
-\operatorname*{E}_{\mathcal{U}}\left[  p_{\ell}\right] \\
&  \leq\sum_{C}\frac{\beta_{C}+\left\vert \beta_{C}\right\vert \varepsilon
}{2^{\left\vert C\right\vert }}-\operatorname*{E}_{\mathcal{U}}\left[
p_{\ell}\right] \\
&  =\operatorname*{E}_{\mathcal{U}}\left[  p_{u}-p_{\ell}\right]
+\varepsilon\sum_{C}\frac{\left\vert \beta_{C}\right\vert }{2^{\left\vert
C\right\vert }}\\
&  =o\left(  1\right)  +\frac{n^{o\left(  1\right)  }}{n^{\Omega\left(
1\right)  }}\\
&  =o\left(  1\right)  .
\end{align*}
Likewise,%
\begin{align*}
\Pr_{x\sim\mathcal{U}}\left[  f\left(  x\right)  \right]  -\Pr_{x\sim
\mathcal{D}}\left[  f\left(  x\right)  \right]   &  \leq\operatorname*{E}%
_{\mathcal{U}}\left[  p_{u}\right]  -\operatorname*{E}_{\mathcal{D}}\left[
p_{\ell}\right] \\
&  =\operatorname*{E}_{\mathcal{U}}\left[  p_{u}\right]  -\sum_{C}\alpha
_{C}\operatorname*{E}_{\mathcal{D}}\left[  C\right] \\
&  \leq\operatorname*{E}_{\mathcal{U}}\left[  p_{u}\right]  -\sum_{C}%
\frac{\alpha_{C}-\left\vert \alpha_{C}\right\vert \varepsilon}{2^{\left\vert
C\right\vert }}\\
&  =\operatorname*{E}_{\mathcal{U}}\left[  p_{u}-p_{\ell}\right]
+\varepsilon\sum_{C}\frac{\left\vert \alpha_{C}\right\vert }{2^{\left\vert
C\right\vert }}\\
&  =o\left(  1\right)  .
\end{align*}

\end{proof}

\section{Discussion\label{DISC}}

We now take a step back, and use our results to address some conceptual
questions about the relativized $\mathsf{BQP}$\ versus $\mathsf{PH}$ question,
the GLN Conjecture, and what makes them so difficult.

The first question is an obvious one. \ Complexity theorists have known for
decades how to prove constant-depth circuit lower bounds, and how to use those
lower bounds to give oracles $A$ relative to which (for example)
$\mathsf{PP}^{A}\not \subset \mathsf{PH}^{A}$ and $\mathsf{\oplus P}%
^{A}\not \subset \mathsf{PH}^{A}$. \ So why should it be so much harder to
give an $A$ relative to which $\mathsf{BQP}^{A}\not \subset \mathsf{PH}^{A}$?
\ What makes this $\mathsf{AC}^{0}$\ lower bound different from other
$\mathsf{AC}^{0}$\ lower bounds?

The answer seems to be that, while we have powerful techniques for proving
that a function $f$ is not in $\mathsf{AC}^{0}$, \textit{all of those
techniques, in one way or another, involve arguing that }$f$\textit{ is not
approximated by a low-degree polynomial.} \ The Razborov-Smolensky technique
\cite{razborov:ac0,smolensky}\ argues this explicitly, while even the random
restriction technique \cite{fss,yao:ph,hastad:book} argues it
\textquotedblleft implicitly,\textquotedblright\ as shown by Linial, Mansour,
and Nisan \cite{lmn}. \ And this is a problem, if $f$ is also computed by an
efficient quantum algorithm. \ For Beals et al.\ \cite{bbcmw} proved the
following in 1998:

\begin{lemma}
[\cite{bbcmw}]\label{bbcmwlem}Suppose a quantum algorithm $Q$\ makes $T$
queries to a Boolean input $X\in\left\{  0,1\right\}  ^{N}$. \ Then $Q$'s
acceptance probability is a real multilinear polynomial $p\left(  X\right)  $,
of degree at most $2T$.
\end{lemma}

In other words, if a function $f$ is in $\mathsf{BQP}$, then for that very
reason, $f$ \textit{has} a low-degree approximating polynomial! \ As an
example, we already saw that the following polynomial $p$, of degree $4$,
successfully distinguishes the forrelated distribution $\mathcal{F}$\ from the
uniform distribution $\mathcal{U}$:
\begin{equation}
p\left(  f,g\right)  :=\frac{1}{N^{3}}\left(  \sum_{x,y\in\left\{
0,1\right\}  ^{n}}f\left(  x\right)  \left(  -1\right)  ^{x\cdot y}g\left(
y\right)  \right)  ^{2}. \label{pfg}%
\end{equation}
Therefore, we cannot hope to prove a lower bound for \textsc{Fourier
Checking}, by any argument that would also imply that such a $p$ cannot exist.

This brings us to a second question. \ If

\begin{enumerate}
\item[(i)] every known technique for proving $f\notin\mathsf{AC}^{0}$ involves
showing that $f$ is not approximated by a low-degree polynomial, but

\item[(ii)] every function $f$ with low quantum query complexity\ \textit{is}
approximated by a low-degree polynomial,
\end{enumerate}

\noindent does that mean there is no hope of solving the relativized
$\mathsf{BQP}$ versus $\mathsf{PH}$\ problem using polynomial-based techniques?

We believe the answer is no. \ The essential point here is that an
$\mathsf{AC}^{0}$ function can be approximated by different \textit{kinds} of
low-degree polynomials. \ For example, Linial, Mansour, and Nisan \cite{lmn}
showed that, if $f:\left\{  0,1\right\}  ^{n}\rightarrow\left\{  0,1\right\}
$\ is in $\mathsf{AC}^{0}$, then there exists a real polynomial $p:\mathbb{R}%
^{n}\rightarrow\mathbb{R}$, of degree $\operatorname*{polylog}n$, such that%
\[
\operatorname*{E}_{x\in\left\{  0,1\right\}  ^{n}}\left[  \left(  p\left(
x\right)  -f\left(  x\right)  \right)  ^{2}\right]  =o\left(  1\right)  .
\]
By comparison, Razborov \cite{razborov:ac0}\ and Smolensky \cite{smolensky}
showed that if $f\in\mathsf{AC}^{0}$, then there exists a polynomial
$p:\mathbb{F}^{n}\rightarrow\mathbb{F}$\ over \textit{any} field $\mathbb{F}$
(finite or infinite), of degree $\operatorname*{polylog}N$, such that%
\[
\Pr_{x\in\left\{  0,1\right\}  ^{n}}\left[  p\left(  x\right)  \neq f\left(
x\right)  \right]  =o\left(  1\right)  .
\]
Furthermore, to show that $f\notin\mathsf{AC}^{0}$, it suffices to show that
$f$ is not approximated by a low-degree polynomial in \textit{any one} of
these senses. \ For example, even though the \textsc{Parity} function has
degree $1$ over the finite field $\mathbb{F}_{2}$, Razborov and Smolensky
showed that over other fields (such as $\mathbb{F}_{3}$), any degree-$o\left(
\sqrt{n}\right)  $ polynomial disagrees with \textsc{Parity}\ on a large
fraction of inputs---and that is enough to imply that \textsc{Parity}%
$\notin\mathsf{AC}^{0}$. \ In other words, we simply need to find a
\textit{type} of polynomial approximation that works for $\mathsf{AC}^{0}%
$\ circuits, but does not work for the \textsc{Fourier Checking} problem. \ If
true, Conjecture \ref{sandwich} (the Low-Fat Sandwich Conjecture) provides
exactly such a type of approximation.

But this raises another question: what is the significance of the
\textquotedblleft low-fat\textquotedblright\ requirement in Conjecture
\ref{sandwich}? \ Why, of all things, do we want our approximating polynomial
$p$ to be expressible as a linear combination of terms, $p\left(  x\right)
=\sum_{C}\alpha_{C}C\left(  x\right)  $, such that\ $\sum_{C}\left\vert
\alpha_{C}\right\vert 2^{-\left\vert C\right\vert }=n^{o\left(  1\right)  }$?

The answer takes us to the heart of what an oracle separation between
$\mathsf{BQP}$ and $\mathsf{PH}$\ would have to accomplish. \ Notice that,
although the polynomial $p$\ from equation (\ref{pfg}) solved the
\textsc{Fourier Checking} problem, it did so only by \textit{cancelling
massive numbers of positive and negative terms,} then representing the answer
by the tiny residue left over. \ Not coincidentally, this sort of cancellation
is a central feature of quantum algorithms. \ By contrast, Theorem
\ref{sandimpgln}\ essentially says that,\ if a polynomial $p$ does
\textit{not} involve such massive cancellations, but is instead more
\textquotedblleft conservative\textquotedblright\ and \textquotedblleft
reasonable\textquotedblright\ (like the polynomials that arise from classical
decision trees), then $p$ cannot distinguish almost $k$-wise independent
distributions from the uniform distribution, and therefore cannot solve
\textsc{Fourier Checking}. \ If Conjecture \ref{sandwich}\ holds, then every
small-depth\ circuit can be approximated, not just by any low-degree
polynomial, but by a \textquotedblleft conservative,\textquotedblright%
\ \textquotedblleft reasonable\textquotedblright\ low-degree polynomial---one
with a bound on the coefficients that prevents massive cancellations.\ \ This
would prove that \textsc{Fourier Checking}\ has no small constant-depth
circuits, and hence that\ there exists an oracle separating $\mathsf{BQP}$
from $\mathsf{PH}$.

This brings us to the fourth and final question:\ how might one prove
Conjecture \ref{sandwich}? \ In particular, is it possible that some trivial
modification of Braverman's proof \cite{braverman}\ would give low-fat
sandwiching polynomials, thereby establishing the GLN Conjecture?

While we cannot rule this out, we believe that the answer is no. \ For
examining Braverman's proof, we find that it combines two kinds of polynomial
approximations of $\mathsf{AC}^{0}$\ circuits: that of Linial-Mansour-Nisan
\cite{lmn}, and that of Razborov \cite{razborov:ac0}\ and Smolensky
\cite{smolensky}. \ Unfortunately, \textit{neither LMN nor Razborov-Smolensky
gives anything like the control over the approximating polynomial's
coefficients that Conjecture \ref{sandwich} demands.} \ LMN simply takes the
Fourier transform of an $\mathsf{AC}^{0}$ function and deletes the high-order
coefficients; while\ Razborov-Smolensky approximates each OR gate by a product
of randomly-chosen linear functions.\ \ Both techniques produce approximating
polynomials with a huge number of monomials, and no reasonable bound on their
coefficients. \ While it is conceivable that those polynomials satisfy the
low-fat condition anyway---because of some non-obvious representation as a
linear combination of terms---certainly neither LMN nor Razborov-Smolensky
gives any idea what that representation would look like. \ Thus, we suspect
that, to get the desired control over the coefficients, one will need more
\textquotedblleft constructive\textquotedblright\ proofs of both the LMN and
Razborov-Smolensky theorems. \ Such proofs would likely be of great interest
to circuit complexity and computational learning theory for independent reasons.

\section{Open Problems\label{OPEN}}

First, of course, prove the GLN Conjecture, or prove the existence of an
oracle $A$ relative to which $\mathsf{BQP}^{A}\not \subset \mathsf{PH}^{A}$ by
some other means. \ A natural first step would be to prove the GLN Conjecture
for the special case of DNFs: as shown in Theorem \ref{glnam}, this would
imply an oracle $A$ relative to which $\mathsf{BQP}^{A}\not \subset
\mathsf{AM}^{A}$. \ We have offered a \$200 prize for the $\mathsf{PH}$\ case
and a \$100 prize for the $\mathsf{AM}$ case.\footnote{See
http://scottaaronson.com/blog/?p=381}

Second, it would be of interest to prove the GLN Conjecture for classes of
functions weaker than (or incomparable with) DNFs: for example, monotone DNFs,
read-once formulas, and read-$k$-times formulas.

Third, can we give an example of a Boolean function $f:\left\{  0,1\right\}
^{n}\rightarrow\left\{  -1,1\right\}  $\ that is well-approximated by a
low-degree polynomial, but \textit{not} by a low-degree low-fat polynomial?
\ Here is a more concrete version of the challenge: let%
\[
\left\Vert f-p\right\Vert :=\operatorname*{E}_{x\in\left\{  0,1\right\}  ^{n}%
}\left[  \left(  f\left(  x\right)  -p\left(  x\right)  \right)  ^{2}\right]
.
\]
Then find a Boolean function $f$ for which

\begin{enumerate}
\item[(i)] there exists a degree-$n^{o\left(  1\right)  }$ polynomial
$p:\mathbb{R}^{n}\rightarrow\mathbb{R}$ such that $\left\Vert f-p\right\Vert
=o\left(  1\right)  $, but

\item[(ii)] there does \textit{not} exist a degree-$n^{o\left(  1\right)  }$
polynomial $q:\mathbb{R}^{n}\rightarrow\mathbb{R}$ such that $\left\Vert
f-q\right\Vert =o\left(  1\right)  $\ and $q$\ can be written as a linear
combination of terms, $q\left(  x\right)  =\sum_{C}\alpha_{C}C\left(
x\right)  $, with $\sum_{C}\left\vert \alpha_{C}\right\vert 2^{-\left\vert
C\right\vert }=n^{o\left(  1\right)  }$.
\end{enumerate}

Fourth, can we give an oracle relative to which $\mathsf{BQP}\not \subset
\mathsf{IP}$? \ What about an oracle relative to which $\mathsf{BQP}%
\neq\mathsf{IP}_{\mathsf{BQP}}$, where $\mathsf{IP}_{\mathsf{BQP}}$\ is the
class of problems that admit an interactive protocol with a $\mathsf{BPP}%
$\ verifier and a $\mathsf{BQP}$\ prover?\footnote{If we let the verifier
transmit unentangled qubits to the prover, then the resulting class
$\mathsf{IP}_{\mathsf{BQP}}^{\left\vert \theta\right\rangle }$\ actually
equals $\mathsf{BQP}$, as recently shown by Broadbent, Fitzsimons, and Kashefi
\cite{bfk}\ (see also Aharonov, Ben-Or, and Eban \cite{abe}). \ It is not
known whether this $\mathsf{IP}_{\mathsf{BQP}}^{\left\vert \theta\right\rangle
}=\mathsf{BQP}$\ result relativizes; we conjecture that it does not.}

Fifth, what other implications does the GLN Conjecture have? \ If we assume
it, can we address other longstanding open questions in quantum complexity
theory, such as those discussed in Section \ref{MOTIVATION}? \ For example,
can we give an oracle relative to which $\mathsf{NP}\subseteq\mathsf{BQP}%
$\ but $\mathsf{PH}\not \subset \mathsf{BQP}$, or an oracle relative to which
$\mathsf{NP}\subseteq\mathsf{BQP}$\ and $\mathsf{PH}$ is infinite?

Sixth, how much can we say about the $\mathsf{BQP}$\ versus $\mathsf{PH}$
question in the unrelativized world? \ As one concrete challenge, can we find
a nontrivial way to \textquotedblleft realize\textquotedblright\ the
\textsc{Fourier Checking} oracle (in other words, an explicit computational
problem that is solvable using \textsc{Fourier Checking})?

Seventh, how far can the gap between the success probabilities of
$\mathsf{FBQP}$\ and $\mathsf{FBPP}^{\mathsf{PH}}$\ algorithms be improved?
\ Theorem \ref{fflb}\ gave a relation for which a quantum algorithm succeeds
with probability $1-c^{-n}$, whereas any $\mathsf{FBPP}^{\mathsf{PH}}%
$\ algorithm succeeds with probability at most $0.99$. \ By changing the
success criterion for \textsc{Fourier Fishing}---basically, by requiring the
classical algorithm\ to output $z_{1},\ldots,z_{n}$ such that $\widehat{f}%
_{1}\left(  z_{1}\right)  ^{2},\ldots,\widehat{f}_{n}\left(  z_{n}\right)
^{2}$\ are distributed \textquotedblleft almost exactly as they would be in
the quantum algorithm\textquotedblright---one can improve the\ $0.99$\ to
$1/2+\varepsilon$\ for any $\varepsilon>0$. \ However, improving the constant
further might require a direct product theorem for $\mathsf{AC}^{0}$\ circuits
solving \textsc{Fourier Fishing}.

\section{Acknowledgments}

I thank Louay Bazzi for reformulating the GLN Conjecture as the Low-Fat
Sandwich Conjecture; and Andy Drucker, Lance Fortnow, and Sasha Razborov for
helpful discussions.

\bibliographystyle{plain}
\bibliography{thesis}

\end{document}